\documentclass[lettersize,journal]{IEEEtran}
\usepackage{amsmath,amsfonts,amssymb,amsthm}
\usepackage{algorithmic}
\usepackage{algorithm}
\usepackage{array}
\usepackage[caption=false,font=normalsize,labelfont=sf,textfont=sf]{subfig}
\usepackage{textcomp}
\usepackage{stfloats}
\usepackage{url}
\usepackage{verbatim}
\usepackage{graphicx}
\usepackage{cite}
\usepackage{hyperref}
\usepackage{tabularx}
\usepackage{multirow} 
\usepackage{booktabs}
\usepackage{fancyhdr}

\newtheoremstyle{ieeeit-head}
  {0pt}                 
  {0pt}                 
  {\normalfont}             
  {\parindent}              
  {\itshape}                
  {:}                       
  {0.5em}                   
  {\thmname{#1}\thmnumber{~#2}\thmnote{~(\itshape #3)}}

\theoremstyle{ieeeit-head}
\newtheorem{definition}{Definition}
\newtheorem{assumption}{Assumption}
\newtheorem{lemma}{Lemma}
\newtheorem{theorem}{Theorem}
\newtheorem{remark}{Remark}

\IEEEoverridecommandlockouts 

\fancypagestyle{arxiv}{%
  \fancyhf{} 
  \fancyhead[C]{%
  \small
    This is the author’s version of the article accepted for publication in \textit{IEEE Transactions on Systems, Man, and Cybernetics: Systems}%
  }
  \fancyfoot[C]{%
    Final version available at \href{https://ieeexplore.ieee.org/document/11220925}{10.1109/TSMC.2025.3623465}%
  }
}

\begin{document}

\title{Prescribed-Time Convergent Distributed Multiobjective Optimization With Dynamic Event-Triggered Communication}

\author{Tengyang Gong, \IEEEmembership{Graduate Student Member, IEEE}, Zhongguo Li, \IEEEmembership{Member, IEEE}, Yiqiao Xu, \IEEEmembership{Member, IEEE}, and Zhengtao Ding, \IEEEmembership{Senior Member, IEEE}
\thanks{Tengyang Gong, Zhongguo Li, Yiqiao Xu, and Zhengtao Ding, are with the Department of Electrical and Electronic Engineering, The University of Manchester, Manchester, M13 9PL, U.K. (e-mail:  tengyang.gong; zhongguo.li; yiqiao.xu; zhengtao.ding@manchester.ac.uk).}
\thanks{\noindent
\copyright~2025 IEEE. Personal use of this material is permitted. 
Permission from IEEE must be obtained for all other uses, in any current or future media, 
including reprinting/republishing this material for advertising or promotional purposes, 
creating new collective works, for resale or redistribution to servers or lists, 
or reuse of any copyrighted component of this work in other works.  
The final version of record is available at:  
\href{https://ieeexplore.ieee.org/document/11220925}{10.1109/TSMC.2025.3623465}.
}
}

\markboth{IEEE TRANSACTIONS ON SYSTEMS, MAN, AND CYBERNETICS: SYSTEMS}%
{Gong \MakeLowercase{\textit{et al.}}: A Sample Article Using IEEEtran.cls for IEEE Journals}


\maketitle

\thispagestyle{arxiv}

\begin{abstract}
This paper addresses distributed constrained multiobjective resource allocation problems (DCMRAPs) in multi-agent networks, where agents face multiple conflicting local objectives under local and global constraints. By reformulating DCMRAPs as single-objective weighted $L_p$ problems, the proposed approach enables distributed solutions without relying on predefined weighting coefficients or centralized decision-making. Leveraging prescribed-time control and dynamic event-triggered mechanisms (ETMs), a novel distributed algorithm is proposed within a prescribed time through sampled communication. Using generalized time-based generators (TBGs), the algorithm provides more flexibility in optimizing solution accuracy and trajectory smoothness without the constraints of initial conditions. Novel dynamic ETMs, integrated with generalized TBGs, improve communication efficiency by adapting to local error metrics and network-based disagreements, while providing enhanced flexibility in balancing solution accuracy and communication frequency. The Zeno behavior is excluded. Validated by Lyapunov analysis and simulation experiments, our method demonstrates superior control performance and efficiency compared to existing methods, advancing distributed optimization across diverse applications.
\end{abstract}


\section{Introduction}
\label{sec:introduction}
\IEEEPARstart{D}{istributed} optimization (DO) has garnered substantial attention over the past two decades due to its scalability, fault tolerance, privacy preservation, and communication efficiency. These features have made it applicable in diverse fields, including machine learning \cite{nedic2020distributedmachine}, robotic coordination control \cite{li2020distributedcoordination}, and energy management systems \cite{zhao2019distributed, Xu2024Power}. A prominent research direction within DO is the distributed resource allocation problem, wherein agents cooperatively allocate resources to optimize global objectives, each agent possessing local objective functions and, in some instances, local constraints \cite{yi2016initialization, chen2022initializationfree, yang2022ADMMnonconvexNonsmooth, zhang2020DualRA}. Continued exploration in DO has yielded numerous algorithms \cite{nedic2009distributedsubgradient, KIA2015Distributed, zhu2019continuoustimegenral, gao2022ETMGT}, typically exhibiting asymptotic or exponential convergence to the optimal solution.

Recent finite-time \cite{chen2017distributedfinitetime, shi2019finite, liu2022finite} and fixed-time \cite{ning2019fixeddistributed, li2021multiobjective, chen2022initializationfree} DO algorithms aim for guaranteed performance and robustness. However, their settling times often depend on initial conditions and control parameters, and they face limitations such as conservative settling time upper bounds and high initial control efforts. To set an a priori bound independent of initial conditions, predefined-time algorithms have been developed \cite{Lin2020predefined, Wang2024predefined,Su2025predefined}. By contrast, prescribed-time control offers smoother control actions and user-specifiable settling times through time-based generators (TBGs). Despite practical studies to mitigate singularities induced by high-gain control \cite{NING2019Practicalfixed, GUO2022distributeddynamic, shi2022predefinedtime, yang2025preETM}, conventional TBGs have limited structural flexibility and introduce a tolerance parameter whose minimization is physically limited by processors, hindering ideal accuracy. A time–space deformation approach is presented in \cite{zhou2025singularity} which avoids high-gain singularities without relying on the tolerance parameter. Moreover, \cite{liu2023multiobjective} proposed a generalized TBG structure that accommodates more diverse TBG forms, enabling smaller errors and encompassing classical TBGs as a special case.

Communication among agents consumes more energy than computation. Different from traditional continuous-time or periodic communication strategies \cite{KIA2015Distributed, zhou2022sampledData}, event-triggered mechanisms (ETMs) efficiently reduce communication frequency by capturing only significant data changes. Event-triggered DO algorithms have been proposed using static ETMs \cite{KIA2015Distributed,gao2022ETMGT, zhao2019distributed, wu2020event, chai2024event} or dynamic ETMs \cite{duDOdynamic, XU2022DETM, Liu2023DONonlinearETM}. Subsequent research has integrated ETMs with finite/fixed-time control \cite{song2022FixedTimeETM, yang2024FixedTimeDynamicETM, JI2024etm}. Although some studies have recently explored dynamic ETMs for prescribed-time DO algorithms \cite{GUO2022distributeddynamic, shi2022predefinedtime, yang2025preETM, yang2025preETMEL}, improving ETM efficiency, particularly through refined triggering thresholds, remains an active research area. This endeavor is particularly challenging as prescribed-time convergence imposes stringent stability conditions, while dynamic ETMs concurrently demand rigorous stability proofs, effective triggering thresholds, and Zeno behavior avoidance. Harmonizing prescribed-time control with dynamic ETMs is therefore mathematically complex, necessitating meticulous parameter design, adaptive feedback laws, and internal variable adjustment rules. 

The majority of the aforementioned literature concentrates on single-objective problems. However, real-world optimization problems often demonstrate multiple conflicting objectives such as economic efficiency, environmental sustainability, and technical robustness \cite{Liu2021LiNing, zhao2019distributed, BARAKAT2020MOPReliability, FONSECA2021Sustainability}, typically addressed through scalarization methods or evolutionary algorithms. Evolutionary algorithms, including ant colony optimization \cite{huang2021AntColony}, genetic algorithms \cite{DENG2022genetic}, and particle swarm optimization \cite{Tian2020Swarm}, are inherently stochastic because they rely on random choices in several steps. Conversely, scalarization methods, like \cite{ YU2022fixedtimeconsensusDO, Liu2021LiNing, liu2023multiobjective}, transform multiple objectives into a single composite function through weighted aggregation. While these approaches enable the application of traditional single-objective DO techniques, they necessitate a sequential two-phase process to allocate weights. Consequently, this often requires the presence of a central decision-maker, which contradicts fully decentralized settings.

To eliminate central authority in multiobjective resource allocation problems (MRAPs), \cite{li2021multiobjective} introduced a fully distributed scheme using an online-constructed weighted $ L_p $ preference index, whereby each agent independently computes its weights and ideal points using only locally available information within a fixed-time paradigm. The version with global prescribed-time convergence is further explored in \cite{zhang_predefined-time_2023}. However, existing approaches do not comprehensively address practical scenarios involving both global and local constraints, alongside the crucial need for communication efficiency and flexibility to diverse applications. 

Motivated by the above discussions, this paper proposes a distributed algorithm for MRAPs with local and global constraints, achieving prescribed-time convergence and sampled communication for improved efficiency and flexibility.

The main contributions are summarized as follows:
\begin{enumerate}
    \item  By dynamically constructing a weighted $L_p$ preference problem in a distributed manner, this study ensures a globally optimal compromise solution to the DCMRAP without any central authority, unlike \cite{YU2022fixedtimeconsensusDO, Liu2021LiNing, liu2023multiobjective}. Moreover, while existing schemes \cite{li2021multiobjective, zhang_predefined-time_2023} neglect local constraints, this study addresses both local and global constraints using the differentiated projection operator and new formulations of weighting coefficients and ideal points. This aligns with the physical limitations, such as output bounds, and enhances real‐world applicability.

    \item  A novel projection-based distributed algorithm with dynamic ETMs is proposed to guarantee prescribed-time convergence. Unlike prior schemes \cite{NING2019Practicalfixed, yang2025preETM, GUO2022distributeddynamic, zhang_predefined-time_2023}, our algorithm offers enhanced flexibility through generalized TBGs, which support diverse functional formulas, enabling optimization of solution accuracy and trajectory smoothness across varied operating conditions. When paired with the proposed dynamic ETMs, it provides a tunable trade‐off between convergence precision and communication frequency, further improving adaptability to practical system requirements.

    \item  Novel dynamic ETMs are deigned to further reduce communication after the prescribed convergence time, in contrast to static ETMs \cite{zhao2019distributed, wu2020event, chai2024event}. Compared to dynamic ETMs \cite{GUO2022distributeddynamic, shi2022predefinedtime, yang2025preETM}, our design combines local error metrics with network-based disagreement to elevate triggering thresholds, while leveraging generalized TBGs to adjust dynamic internal variables. This integration guarantees prescribed‐time convergence, excludes Zeno behavior, and enhances overall performance. Consequently, the proposed ETMs are highly suitable for resource-constrained applications, such as wireless sensor networks and distributed robotics, where communication efficiency is crucial.
     
    \item Rigorous Lyapunov analysis and detailed microgrid simulations confirm prescribed-time convergence and improved communication efficiency, outperforming existing methods.

\end{enumerate}

The remainder of this paper is structured as follows: Section \ref{section2} presents fundamental preliminaries and formulates DCMRAPs. Section \ref{section3} details the distributed algorithm along with the convergence analysis. Section \ref{section4} demonstrates various experiments. Finally, Section \ref{section5} concludes this paper.

\section{PRELIMINARIES AND PROBLEM FORMULATION}\label{section2}
\subsection{Notations}
The sets of real numbers, non-negative real numbers, positive real numbers, $n$-dimensional real vectors and $m\times n$-dimensional matrices are denoted as $\mathbb{R}$, $\mathbb{R}_{+}$, $\mathbb{R}_{++}$, $\mathbb{R}^{n}$, $\mathbb{R}^{m\times n}$, respectively. For a matrix $A=\left[a_{ij}\right]\in \mathbb{R}^{m\times n}$, $a_{ij}$ represents the element in the $i$th row and $j$th column. $\mathbf{0}_{n}$, $\mathbf{1}_{n}$ and $I_{n}$ denote $n$-dimensional vectors with all entries equal to zero, all entries equal to one, and an ${n\times n}$ identity matrix, respectively. Let $\nabla f(x)$ be the gradient of a function $f(x):\mathbb{R}^n\to \mathbb{R}$ at the point $x\in \mathbb{R}^{n}$. $\mathbf{dom}\ f$ denotes the feasible domain of $f$. $\left \| \cdot \right \|$ denotes the Euclidean norm. $\times$ denotes the Cartesian product. Let $\lambda _{i}(A)$ be the $i$th smallest eigenvalue of the matrix $A$ with $\lambda _{1}(A) \leq \lambda _{2}(A) \leq \cdots \leq \lambda _{n}(A)$. Let $\mathrm{int}(\Omega)$ and $\partial \Omega$ denote the sets of interiors and boundaries of set $\Omega$, respectively.

\subsection{Graph Theory}
Let $\mathcal{G}=(\mathcal{V},\mathcal{E})$ denote a graph with the set of agents $\mathcal{V}=\left \{ 1,\dots,N \right \}$ and the set of edges $\mathcal{E}\subseteq \mathcal{V}\times \mathcal{V}$. A graph is undirected if and only if $(i,j)\in\mathcal{E}$ and $(j,i)\in\mathcal{E}$ hold simultaneously. Let $A=\left[a_{ij}\right]\in\mathbb{R}^{N\times N}$ be adjacency matrix of $\mathcal{G}$, where $a_{ij}>0$ if $\left(j,i\right)\in\mathcal{E}$ and $a_{ij}=0$ otherwise. Define the Laplacian matrix $L=\left[l_{ij}\right]\in\mathbb{R}^{N\times N}$ of graph $\mathcal{G}$ as $l_{ii} =\sum_{j=1,j\ne i}^{N}a_{ij}$ and $l_{ij} =-a_{ij},i\ne j$. $0$ is a simple eigenvalue of $L$ with the associated eigenvector $\mathbf{1}_N$ and all other eigenvalues are positive.
\begin{lemma}[Orthogonal Transformation \cite{yi2018doubleorder}] \label{orthogonaltransformation}
    There exists an orthogonal matrix $Q=\left [q_1 \ Q_2 \right ]\in \mathbb{R}^{N\times N}$ with $q_1 = \frac{1}{\sqrt{N}}\mathbf{1}_N$ such that
    \begin{equation}
         L = \left [q_1 \ Q_2 \right ]\begin{bmatrix}
                                    0& \\ 
                                     &\Lambda 
                                    \end{bmatrix}\begin{bmatrix}
                                    q_1^\top \\ 
                                     Q_2^\top 
                                    \end{bmatrix},
    \end{equation}
    where $\Lambda = \mathrm{diag}(\lambda_2(L),\dots,\lambda_N(L))$.
\end{lemma}

\subsection{Convexity and Projection Operator}
\begin{definition}[Strong Convexity]\label{strongconvex}
    A function $f:\mathbb{R}^n\to\mathbb{R}$ is $m$-strongly convex if $\mathbf{dom}\ f$ is a convex set and if $\forall x_1,x_2\in\mathbf{dom}\ f$, $\exists \ m>0$ such that $
        f\left(x_2\right)\geq f\left(x_1\right)+{\nabla f\left(x_1\right)}^\top\left(x_2-x_1\right)+\frac{m}{2}\left \|x_2-x_1 \right \|^2.$  
\end{definition}
\begin{lemma}[\cite{ruszczynski2011nonlinear}]\label{lemmaStronglyConvex}
    If $f:\mathbb{R}^n\rightarrow\mathbb{R}$ is $m$-strongly convex, then $\forall x_1,x_2\in\mathbf{dom}\ f$, 
     \begin{equation}
        \left ( \nabla f(x_2)-\nabla f(x_1)\right )^\top \left ( x_2 -x_1 \right ) \geq m \left \| x_2 -x_1 \right \|^2. 
    \end{equation}
\end{lemma}
 \begin{definition}[Normal Cone and Tangent Cone \cite{ruszczynski2011nonlinear}]
         Let $\Omega \subset \mathbb{R}^{n}$ be a closed, convex and nonempty set. 
         The unit normal cone and tangent cone at $x \in \Omega$ are defined as $ n_{\Omega}(x) \triangleq \left \{ z \in \mathbb{R}^{n} \mid  \left \|z \right \|=1, z^\top(y-x)\le 0, \forall y \in \Omega  \right \}$ and $ T_{\Omega}(x) \triangleq \left \{ \displaystyle \lim_{k\to \infty } \frac{x_k - x}{\tau_{k}}\mid  x_{k}\in \Omega, x_{k}\rightarrow x, \tau_{k}>0, \tau_{k}\rightarrow0  \right \}$.
    \end{definition}
    
For a closed convex set $\Omega$, a point $x\in \Omega$ and a direction $v \in \mathbb R^n $, the differentiated projection operator is defined as $P_{T_\Omega(x)}(v)= \lim_{\varepsilon  \rightarrow 0}\frac{P_\Omega(x+\varepsilon v)-x}{\varepsilon}$, where $P_{\Omega}(x) = \mathrm{arg}\min_{y \in \Omega}\left \|y-x \right \|$ is the Euclidean projection operator.
    
    \begin{lemma}[\cite{BROGLIATO2006ontheequivalence,yi2016initialization}] \label{propertiesDifferentialprojection}
        The differentiated projection operator has the following properties: 
        i) If $x \in \mathrm{int}(\Omega)$, then $P_{T_\Omega(x)}(v)=v$ ; ii) $x\in \partial \Omega$, and $\displaystyle\max_{z \in n_\Omega(x)}v^\top z\leq0$, then $P_{T_\Omega(x)}(v)=v$; iii) $x\in \partial \Omega$, and $\displaystyle\max_{z \in n_\Omega(x)}v^\top z\geq0$, then $P_{T_\Omega(x)}(v)=v-v^\top z^*z^*$, where $z^*=\mathrm{arg}\displaystyle\min_{z \in n_\Omega(x)}v^\top z$. 
    \end{lemma}       
\subsection{Prescribed-Time Convergence} 
\begin{definition}[\cite{liu2023multiobjective}]
    A system $\dot{x}(t)=T(t,t_{\rm pre})f\left (x(t) \right )$ is said to achieve prescribed-time approximate convergence at $t_{\rm pre} $ if for any $x(0)$, there exists $0<\epsilon=\epsilon(x(0))$ such that the following three conditions hold
    \begin{equation}
        \left\{\begin{matrix}
            \lim_{t\rightarrow t_{\rm pre +}}\left \|x(t) \right \|\leq \epsilon\\ 
            \left \|x({t}^{\prime}) \right \|\leq \epsilon, \forall {t}^{\prime}>t_{\rm pre }\\ 
            \lim_{t\rightarrow \infty}\left \|x(t) \right \|=0
    \end{matrix}\right.
    \end{equation}
    where $t_{\rm pre}$ is a user-assignable time without dependence on initial states, and $ T(t,t_{\rm pre}):\mathbb{R}_+\times \mathbb{R}_{++}\to \mathbb{R}$ is a TBG.
\end{definition}

Given a TBG defined as  
\begin{equation}\label{definitionTBG}
    T(t,t_{\rm pre}) = \frac{\mathrm{d}\gamma (t,\sigma )}{\mathrm{d}t},
\end{equation}
where $\gamma (t,\sigma )$ satisfies 
i) $ 0 < \sigma \ll 1 $;  
ii) $ \lim_{\sigma \to 0_+} \left[ \gamma (t_{\rm pre+}, \sigma) - \gamma (0, \sigma) \right] = +\infty $;  
iii) $ \gamma (t, \sigma) - \gamma (t_{\rm pre+}, \sigma) \geq 0, \, \forall t > t_{\rm pre+} $;  
iv) $ \lim_{t \to +\infty} \left[ \gamma (t, \sigma) - \gamma (0, \sigma) \right] = +\infty $.  The following lemma establishes the condition for prescribed-time convergence.

\begin{lemma}[\cite{liu2023multiobjective}]\label{prescribedcondition}
    Suppose that there is a Lyapunov function $V(x):\mathbb{R}\rightarrow\mathbb{R}$ satisfying that $\exists \xi >0$ such that $V(x)\geq \xi \left \| x \right \|^2$ and
    \begin{equation}
        \dot{V}\left ( x(t)\right ) \leq - \zeta  T(t,t_{\rm pre})V\left ( x(t)\right )
    \end{equation}
    where $T(t,t_{\rm pre})$ is defined in (\ref{definitionTBG}), and then the origin of system $\dot{x}(t)=T(t,t_{\rm pre})f\left (x(t) \right )$ achieves prescribed-time approximate convergence at $t_{\rm pre}$. In addition, $\epsilon=\sqrt{e^{-\zeta  \left ( \gamma (t_{\rm pre +})-\gamma(0)\right )}\frac{V\left ( x(0)\right )}{\xi}}$.
\end{lemma}

\subsection{Problem Formulation}
Consider a group of $N$ agents collaboratively solving a DCMRAP with local and global constraints. Each agent $i$ has $K$ conflicting objectives:
\begin{equation}\label{problemMultiobjectiveresourceallocation}
    \begin{aligned}
        \min_{\boldsymbol{x}} &\ \left \{ f_i^1(x_i),\dots,f_i^K(x_i)\right \}, \forall i \in \mathcal{V}, \\
        {\rm s.t.}&\  \sum_{i=1}^{N}x_i=D,\quad x_i\in \Omega_i,\\
    \end{aligned}
\end{equation}
where $\boldsymbol{x}=\left [x_1,\dots,x_N\right ]^\top  \in \mathbb{R}^{N}$ is the global decision variable, $x_i \in \mathbb{R}$ is the local decision variable of agent $i$, $f_i^k(x_i):\ \mathbb{R} \rightarrow \mathbb{R}$ is the $k$th objective function of agent $i$ for $k \in \mathcal{K}$, $\Omega_i$ is the local constraint in the form of a convex compact set, and $D=\sum_{i=1}^{N}d_i$ is the total demand, with $d_i$ representing the local demand of agent $i$. Letting $\boldsymbol{\Omega} = \Omega_1\times,\dots,\times\Omega_N$, the feasible domain is $\boldsymbol{\mathcal{X}}=\left \{ \boldsymbol{x} \in \boldsymbol{\Omega} \mid  \mathbf{1}_N^\top  \boldsymbol{x}=D\right \}$. Note that $x_i, f_i^k, d_i$ and $\Omega_i$ are known only to the local agent $i$.  
\begin{assumption}\label{MultiAssumptionGraph}
    The graph is undirected and connected. 
\end{assumption}
\begin{assumption}\label{assumptionMultiConvex}
    Each local objective function $f_i^k$ is continuously differentiable and $m_i^k$-strongly convex.
\end{assumption}
\begin{assumption}[Slater's condition]\label{assumptionMultiSlater}
    There exists at least an interior point $x_i \in \mathrm{int}(\Omega_i)$ such that the global equality constraint $\sum_{i=1}^{N}x_i=D$ for all $i \in \mathcal{V}$ and $k \in \mathcal{K}$.
\end{assumption}

To solve problem (\ref{problemMultiobjectiveresourceallocation}) in a distributed manner, inspired by \cite{li2021multiobjective}, the weighted $L_p$ preference problem is given by 
\begin{equation}\label{problemREMultiobjectiveresourceallocation}
    \begin{aligned}
        \min_{\boldsymbol{x}} &\ U(\boldsymbol{x},\hat{\boldsymbol{x}}^{*},\boldsymbol{\omega}^{*})= \displaystyle \sum_{i=1}^{N} u_i(x_i,\hat{x}_i^{*},\omega_i^{*})\\
        &\ =\displaystyle \sum_{i=1}^{N}\left ( \displaystyle \sum_{k=1}^{K} \omega_i^{k*}\left ( f_i^k(x_i)- f_i^k(\hat{x}_i^{k*})\right )^p  \right)^{\frac{1}{p}},\\
        & \ 1\leq p< \infty,\\
        {\rm s.t.}&\  \sum_{i=1}^{N}x_i=D,\quad x_i\in \Omega_i,\\
    \end{aligned}
\end{equation}
where $u_i(x_i,\hat{x}_i^{*},\omega_i^{*})= \left ( \sum_{k=1}^{K} \omega_i^{k*}\left ( f_i^k(x_i)- f_i^k(\hat{x}_i^{k*})\right )^p  \right)^{\frac{1}{p}}$; $\boldsymbol{\omega}^{*}=\left [\omega_1^{*},\dots,\omega_N^{*}\right ]^\top  \in \mathbb{R}^{KN}_{++}$ collects the local positive weighting vectors $\omega_i^{*}=\left [\omega_i^{1*},\dots,\omega_i^{K*}\right ]^\top  \in \mathbb{R}^{K}_{++}$ specified by each agent $i$, with $\sum_{k=1}^{K}\omega_i^{k*}=1$; $f_i^k(\hat{x}_i^{k*})$ denotes the ideal point of agent $i$ for objective $k$; and $\hat{\boldsymbol{x}}^{*}=\left [\hat{x}_1^{*},\dots,\hat{x}_N^{*}\right ]^\top  \in \mathbb{R}^{KN}$ concatenates each agent’s ideal decision vector $\hat{x}_i^{*}=\left [\hat{x}_i^{1*},\dots,\hat{x}_i^{K*}\right ]^\top  \in \mathbb{R}^{K}$, with $\hat{x}_i^{k*}$ determined by the local optimization:
\begin{equation}\label{ideal}
    \hat{x}_i^{k*} = \arg\min_{\hat{x}_i^{k} \in \Omega_i} f_i^k(\hat{x}_i^k).
\end{equation}

\begin{lemma}\label{lemmaconvexu}
    $u_i(x_i,\hat{x}_i^{*},\omega_i^{*})$ defined in (\ref{problemREMultiobjectiveresourceallocation}) is strictly convex for all $i \in \mathcal{V}$. 
\end{lemma}
\begin{IEEEproof}
    Because $p$-norm function is convex and monotonically increasing, $u_i(f_i^k(x_i))$ is a strictly convex and monotonically increasing function on $f_i^k(x_i)\geq f_i^k(\hat{x}_i^{k*})$. Also, all the objective functions $f_i^k(x_i)$ are strongly convex on $x_i$. Thus, $u_i(x_i,\hat{x}_i^{*},\omega_i^{*})$ is strictly convex on $x_i$. 
\end{IEEEproof}
\begin{assumption}\label{assumptionconvexu}
  Each local preference index $u_i(x_i,\hat{x}_i^{*},\omega_i^{*})$ is continuously differentiable and $\varpi_i$-strongly convex 
\end{assumption}

As described in \cite{li2021multiobjective}, the weight coefficients can be selected based on the relative importance of objective $k$ compared to the total objective value, formulated as:
\begin{equation}\label{8}
    \omega_i^{k*} = \frac{\left |f_i^k(\bar{x}_i^{k*}) \right |}{\sum_{j=1}^{K}\left |f_i^j(\bar{x}_i^{j*}) \right |},
\end{equation}
where $\bar{x}_i^{k*}\in \mathbb{R}$ denotes the optimizer for the summation of each type of objective function subject to both the local and global constraints, that is, for all $k \in \mathcal{K}$,
\begin{equation}\label{problemforkthobjective}
    \begin{aligned}
        \bar{x}_i^{k*} = \arg\min_{\bar{x}_i^k, \forall i \in \mathcal{V}} &\ \sum_{i=1}^{N} f_i^k(\bar{x}_i^k), \quad k \in \mathcal{K}\\
        {\rm s.t.}&\  \sum_{i=1}^{N}\bar{x}_i^k=D,\quad \bar{x}_i^k\in \Omega_i.\\
    \end{aligned}
\end{equation}

To this end, the DCMRAP (\ref{problemMultiobjectiveresourceallocation}) is transformed into a single-objective optimization formulation (\ref{problemREMultiobjectiveresourceallocation}), enabling resolution through DO techniques.

\section{MAIN RESULTS}\label{section3}
\subsection{Event-Triggered Prescribed-Time Distributed Algorithm}
Based on the above discussion, a DO algorithm is proposed for the problem (\ref{problemREMultiobjectiveresourceallocation}) based on prescribed-time control and dynamic ETMs. Firstly, to solve (\ref{problemforkthobjective}) for all $k \in \mathcal{K}$, a distributed protocol is given as 
\begin{subequations}\label{solutionforsubproblem}
    \begin{align}
        \dot{\bar x}_{i}^{k} &= P_{T_{\Omega_i}(\bar{x}_i^k)}\left (  T_1(t,t_{\rm pre 1})\left ( y_i^k-\nabla f_i^k\left (\bar{x}_i^k\right )\right )\right ),\\
        \dot{y}_i^k & =   T_1(t,t_{\rm pre 1})\left ( -\sum_{j=1}^{N}a_{ij}\left (  \bar{y}_i^k-\bar{y}_j^k\right) - z_i^k + d_i-\bar{x}_i^k\right ),\\
        \dot{z}_{i}^{k} & =  T_1(t,t_{\rm pre 1})  \sum_{j=1}^{N}a_{ij} \left ( \bar{y}_{i}^{k}-\bar{y}_j^k \right),\\
        \omega_i^{k} &= \frac{\left |f_i^k(\bar{x}_i^{k}) \right |}{\sum_{j=1}^{K}\left |f_i^j(\bar{x}_i^{j}) \right |},
    \end{align}
\end{subequations}
where $T_1(t,t_{\rm pre 1})$ is the TBG defined in (\ref{definitionTBG}), with $t_{\rm pre 1}$ as the prescribed time for $\bar{x}_i^k, \forall i \in \mathcal{V}, \forall k \in \mathcal{K}$, to converge to the optimal solutions and for determining weighting coefficients; $\omega_i^{k}$ is the estimate of the weighting coefficient $\omega_i^{k*}$; $y_i^k, z_i^k\in \mathbb{R}$ are auxiliary variables with $z_i^k(0)=0$; and $\bar{y}_i^k(t) = y_i^k(t_{i,k}^\ell), t\in \left [t_{i,k}^\ell, t_{i,k}^{\ell+1}\right)$, with $\left \{ t_{i,k}^\ell \right \}$ being the local triggering time sequence for the $k$th objective of $i$th agent. The corresponding dynamic ETM is designed as 
\begin{equation}\label{thresholdETM1}
      t_{i,k}^{\ell+1} = \inf_{t\geq t_{i,k}^\ell}\Bigg\{t \mid  \alpha_i^k \left(\frac{1}{2\varsigma_i^k}+l_{ii}\right)\left \|e_{y,i}^k(t) \right \|^2 \geq \mathcal{H}_i^k(t)  \Bigg\},
\end{equation}
where $e_{y,i}^k(t) = \bar{y}_i^k(t) - y_i^k(t)$ is the local state error, $\mathcal{H}_i^k(t) = \eta_i^k(t) + \frac{\alpha_i^k \beta_i^k}{2} \bar{q}^k_i(t)$ is the dynamic threshold, $\bar{q}^k_i(t) = \frac{1}{2} \sum_{j=1}^N a_{ij} \left \|\bar{y}_i^k(t)-\bar{y}_j^k(t) \right \|^2 \geq 0$, and $\eta_i^k$ is an internal dynamic variable that evolves based on  
\begin{equation}\label{dynamicETM}
\begin{aligned}
    \dot{\eta}_i^k(t) &= T_1(t,t_{\rm pre 1})\left ( -\phi_i^k\eta_i^k(t) \right.\\& \left.
    \quad -\delta_i^k \left ( (\frac{1}{2\varsigma_i^k}+l_{ii})\left \|e_{y,i}^k(t) \right \|^2 - \frac{\beta_i^k}{2} \bar{q}^k_i(t) \right )\right ),
\end{aligned}
\end{equation}
where $\phi_i^k>0, \delta_i^k\in \left ( 0,1\right ], \alpha_i^k> \frac{1-\delta_i^k}{\phi_i^k}$, ${\eta}_i^k(0)>0$, $\beta_i^k \in (0,1)$ and $0<\varsigma_i^k < \min\left \{ \frac{m^k_{\min}}{3}, \frac{\lambda_2(L)(1- \beta_{\max}^k)}{6\psi_y^k}\right \}$ are control parameters with $m^k_{\min} = \min_{i\in\mathcal{V}}{m_i^k}$, $\beta_{\max}^k = \max_{i \in \mathcal{V}} \beta_i^k$ and $\psi_y^k$ specified in \eqref{CE1}. 

To solve problem \eqref{ideal}, the ideal point seeking algorithm is 
\begin{equation}\label{idealpoint}
    \dot{\hat{x}}_i^k = P_{T_{\Omega_i}(\hat{x}_i^k)}\left ( -T_2(t,t_{\rm pre 2}) \nabla f_i^k\left (\hat{x}_i^k\right )\right ),
\end{equation}
where $ T_2(t, t_{\rm pre 2}) $ is the TBG defined in (\ref{definitionTBG}), and $ t_{\rm pre 2} $ denotes the prescribed time for convergence of $\hat{x}_i^k$ and determining the ideal points.

Then, the compromised solution for (\ref{problemREMultiobjectiveresourceallocation}) is obtained by
\begin{subequations}\label{solutionforproblem}
    \begin{align}
        \dot{x}_{i} &= P_{T_{\Omega_i}(x_i)}\left (  T_3(t,t_{\rm pre 3})\left ( \nu_i-\nabla u_i(x_i,\hat{x}_i,\omega_i)\right )\right ),\\
        \dot{\nu}_i & =   T_3(t,t_{\rm pre 3})\left ( -\sum_{j=1}^{N}a_{ij}\left (  \bar{\nu}_i-\bar{\nu}_j\right)- \mu_i+d_i-x_i\right ),\\
        \dot{\mu}_{i} & =  T_3(t,t_{\rm pre 3}) \sum_{j=1}^{N}a_{ij} \left ( \bar{\nu}_{i}-\bar{\nu}_j \right),
    \end{align}
\end{subequations}
where $T_3(t,t_{\rm pre 3})$ is the TBG, with $t_{\rm pre 3}$ as the prescribed time for $\boldsymbol{x}$ to converge to the optimal solution; $\nu_i, \mu_i\in \mathbb{R}$ are auxiliary variables with $\mu_i(0) = 0$; $\bar{\nu}_i(t) = \nu_i(t_{i}^\ell), t\in \left [t_{i}^\ell, t_{i}^{\ell+1}\right)$; and $\left \{ t_{i}^\ell \right \}$ is the local triggering time sequence for $i$th agent. The corresponding dynamic ETM is designed as
\begin{equation}\label{thresholdETM2}
       t_{i}^{\ell+1} = \inf_{t\geq t_{i}^\ell}\Bigg\{t \mid  \alpha_i \left(\frac{1}{2\varsigma_i}+l_{ii}\right)\left \|e_{\nu,i}(t) \right \|^2 \geq \mathcal{H}_i(t)  \Bigg\},
\end{equation}
where $e_{\nu,i}(t) = \bar{\nu}_i(t) - \nu_i(t)$ is the local state error, $\mathcal{H}_i(t) = \eta_i(t) + \frac{\alpha_i \beta_i}{2} \bar{q}_i(t)$ is the dynamic threshold, $\bar{q}_i(t) = \frac{1}{2} \sum_{j=1}^N a_{ij} \left \|\bar{\nu}_i(t)-\bar{\nu}_j(t) \right \|^2 \geq 0$, and $\eta_i$ is updated by 
\begin{equation}\label{dynamicETM2}
    \begin{aligned}
    \dot{\eta}_i(t) &= T_3(t,t_{\rm pre 3})\left ( -\phi_i\eta_i(t) \right.\\& \left.
    \quad -\delta_i \left ( (\frac{1}{2\varsigma_i}+l_{ii})\left \|e_{\nu,i}(t) \right \|^2 - \frac{\beta_i}{2} \bar{q}_i(t) \right )\right ),
\end{aligned}
\end{equation}
where $\phi_i>0, \delta_i\in \left ( 0,1\right ], \alpha_i> \frac{1-\delta_i}{\phi_i}$, ${\eta}_i(0)>0$, $\beta_i \in (0,1)$ and $0<\varsigma_i < \min\left \{ \frac{\varpi_{\min} }{3}, \frac{\lambda_2(L)(1- \beta_{\max})}{6\psi_y}\right \}$ with $\varpi_{\min} = \min_{i\in\mathcal{V}}{\varpi_i}$, $\beta_{\max} = \max_{i \in \mathcal{V}} \beta_i$ and $\psi_y$ defined in \eqref{CE3}.

\begin{remark}
    While equations (\ref{solutionforsubproblem})--(\ref{dynamicETM}) and (\ref{solutionforproblem})--(\ref{dynamicETM2}) share the same prescribed-time control structure and ETM strategy, they serve distinct purposes: the former constructs the weighting coefficients online for each $k \in \mathcal{K}$ within $t_{\rm pre 1}$, while the latter integrates these results to compute a compromised solution for the DCMRAP at $t_{\rm pre3}$. Notably, equation (\ref{idealpoint}) operates independently of ETMs, as it involves only local computations.
\end{remark}
\begin{remark}
  Equations (\ref{solutionforsubproblem}a)–(\ref{solutionforsubproblem}c) implement an event-triggered prescribed-time protocol to solve the constrained subproblem \eqref{problemforkthobjective} in a fully distributed manner, yielding the optimizers $\bar{x}_i^{k*}$. Equation (\ref{solutionforsubproblem}d) online computes each agent’s local weighting coefficients via \eqref{8}, thereby biasing the aggregate objective in \eqref{problemREMultiobjectiveresourceallocation} toward the agent’s prioritized criteria and aligning the global allocation with heterogeneous preferences.
\end{remark}

\subsection{Convergence Analysis}
Before proving the convergence of the algorithm, we give the properties of the designed triggering mechanisms (\ref{dynamicETM}) and (\ref{dynamicETM2}), which are used in the proofs later.

\begin{lemma}[Positivity of $\eta_i^k(t)$ and $\eta_i(t)$]\label{lemmapositive}
    With properly designed parameters $\phi_i^k, \delta_i^k, \alpha_i^k, \beta_i^k, \varsigma_i^k$ for $\eta_i^k(t)$, and $\phi_i, \delta_i, \alpha_i, \beta_i, \varsigma_i$ for $\eta_i(t)$, both variables satisfy $\eta_i^k(t) > 0$ and $\eta_i(t) > 0$ for all $i \in \mathcal{V}$, $k \in \mathcal{K}$.
\end{lemma}

\begin{IEEEproof}
    From (\ref{thresholdETM1}) and (\ref{dynamicETM}), we have $\dot{\eta}_i^k(t) > -(\phi_i^k + \delta_i^k/\alpha_i^k) T_1(t,t_{\text{pre} 1})\eta_i^k(t)$. Together with $\eta_i^k(0)>0$, we have
    \begin{equation}\label{etageq0}
        \eta_i^k(t) > \eta_i^k(0) e^{ -(\phi_i^k + \delta_i^k/\alpha_i^k) \gamma_1(t,\sigma) }>0,
    \end{equation}
    for all $t\geq0$. It can likewise be inferred that the variable $\eta_i(t)$ remains positive for all $t \geq 0$.
\end{IEEEproof}

We first analyze the convergence of the module \eqref{solutionforsubproblem}--\eqref{dynamicETM}. The distributed dynamical system in (\ref{solutionforsubproblem}) can be rewritten in a compact form of 
\begin{subequations}\label{solutionforsubproblemCompact}
\begin{align}
        \dot{\bar{\boldsymbol{x}}}^{k} &= P_{T_{\boldsymbol{\Omega}}(\bar{\boldsymbol{x}}^k)}\left (  T_1(t,t_{\rm pre 1})\left ( \boldsymbol{y}^k-\nabla f^k\left (\bar{\boldsymbol{x}}^k\right )\right )\right ),\\
        \dot{\boldsymbol{y}}^k & =   T_1(t,t_{\rm pre 1})\left ( -\ L\boldsymbol{\bar{y}}^k- \boldsymbol{z}^k+\boldsymbol{d}-\bar{\boldsymbol{x}}^k\right ),\\
        \dot{\boldsymbol{z}}^{k} & =  T_1(t,t_{\rm pre 1})  L\boldsymbol{\bar{y}}^k,
\end{align} 
\end{subequations}
where $\boldsymbol{\Omega} = \Omega_1\times\dots\times\Omega_N$, $\bar{\boldsymbol{x}}^k = \left [ \bar{x}_1^k,\dots,\bar{x}_N^k\right ]^\top \in \mathbb{R}^N$, $\nabla f^k\left (\bar{\boldsymbol{x}}^k\right )=\left [ \nabla f_1^k\left (\bar{x}_1^k\right ),\dots,\nabla f_N^k\left (\bar{x}_N^k\right )\right ]^\top \in \mathbb{R}^N$, $\boldsymbol{y}^k = \left [ y_1^k,\dots,y_N^k\right ]^\top \in \mathbb{R}^N$, $\boldsymbol{\bar{y}}^k = \left [ \bar{y}_1^k,\dots,\bar{y}_N^k\right ]^\top \in \mathbb{R}^N$, $\boldsymbol{z}^k = \left [ z_1^k,\dots,z_N^k\right ]^\top \in \mathbb{R}^N$ and $\boldsymbol{d} = \left [ d_1,\dots,d_N\right ]^\top \in \mathbb{R}^N$.
\begin{lemma}[Optimality]\label{lemmaSubOptimality}
    Under Assumptions \ref{MultiAssumptionGraph}, \ref{assumptionMultiConvex} and \ref{assumptionMultiSlater}, for any bounded initial points $\bar{x}_i^k(0)\in \Omega_i, \forall i \in \mathcal{V}$, if $\left (\bar{\boldsymbol{x}}^{k*}, \boldsymbol{y}^{k*}, \boldsymbol{z}^{k*}\right )$ is the equilibrium of (\ref{solutionforsubproblemCompact}), then $\bar{\boldsymbol{x}}^{k*}$ is the optimal solution to problem (\ref{problemforkthobjective}).
\end{lemma}
\begin{IEEEproof}
From (\ref{solutionforsubproblem}), we have $\dot{\bar x}_{i}^{k}\in T_{\Omega_i}(\bar{x}_i^k)$ when $t\geq 0$. Based on \cite{Aubin1984differentialinclusions}, for any bounded $\bar{x}_i(0)\in \Omega_i$, $\bar{x}_i(t)\in \Omega_i, \forall t \geq 0, i \in \mathcal{V}$ hold. Then we have
\begin{subequations}\label{EQforsubproblemCompact}
    \begin{align}
        \mathbf{0}_N &= P_{T_{\Omega}(\bar{\boldsymbol{x}}^{k*})}\left (  \boldsymbol{y}^{k*}-\nabla f^k\left (\bar{\boldsymbol{x}}^{k*}\right )\right ),\\
        \mathbf{0}_N & =    - L\boldsymbol{y}^{k*}- \boldsymbol{z}^{k*}+\boldsymbol{d}-\bar{\boldsymbol{x}}^{k*},\\
        \mathbf{0}_N & =    L\boldsymbol{y}^{k*}.
    \end{align}
\end{subequations}
According to Lemma \ref{propertiesDifferentialprojection}, (\ref{EQforsubproblemCompact}a) satisfies $0 \in y_i^{k*}-\nabla f_i^k\left (\bar{x}_i^{k*}\right )+N_{\Omega_i}(\bar{x}_i^{k*})$. Moreover, based on the stochastic property of the Laplacian matrix $L$, one obtains $y_1^{k*} = y_2^{k*}=\cdots=y_N^{k*}$ from (\ref{EQforsubproblemCompact}c). Since $z^k_i(0)=0, \forall i\in \mathcal{V}$, $\sum_{i=1}^{N}\dot{z}_i^k = T_1(t,t_{\rm pre 1})(\mathbf{1}_N^\top L\boldsymbol{\bar{y}}^{k})=0 $ holds. This implies that $\sum_{i=1}^{N}z_i^k(t)\triangleq\sum_{i=1}^{N}z_i^k(0)\triangleq0$. By multiplying both sides of (\ref{EQforsubproblemCompact}b) by $\mathbf{1}_N^\top $, one obtains $\sum_{i=1}^{N} \bar{x}_i^{k*} = \sum_{i=1}^{N} d_i$. In conclusion, the equilibrium point $\left (\bar{\boldsymbol{x}}^{k*}, \boldsymbol{y}^{k*}, \boldsymbol{z}^{k*}\right )$ satisfies the Karush-Kuhn-Tucker (KKT) optimality conditions, and thus $\bar{\boldsymbol{x}}^{k*}$ is the optimal global decision variable of problem (\ref{problemforkthobjective}).
\end{IEEEproof}
\begin{lemma}\label{lemmasub}
    Under Assumptions \ref{MultiAssumptionGraph}, \ref{assumptionMultiConvex} and \ref{assumptionMultiSlater}, for all $k \in \mathcal{K}$, the proposed dynamical system (\ref{solutionforsubproblem}) with the dynamic ETM (\ref{thresholdETM1})-(\ref{dynamicETM}) solves the optimization problem (\ref{problemforkthobjective}) at the prescribed time $t_{\rm pre1}$, and the Zeno behavior is excluded, that is,
    \begin{equation}\label{CE1}
        \left\{\begin{array}{@{}l}
            \lim_{t\rightarrow t_{\rm pre 1+}}\left \|\bar{\boldsymbol{x}}^k(t)-\bar{\boldsymbol{x}}^{k*} \right \|\leq \sqrt{e^{-\frac{\kappa_1}{\theta_2^k} \left ( \gamma_1 (t_{\rm pre1 +})-\gamma_1(0)\right )}\frac{V^k\left ( 0\right )}{\theta_2^k}},\\ 
            \left \|\bar{\boldsymbol{x}}^k(t)-\bar{\boldsymbol{x}}^{k*} \right \|\leq \sqrt{e^{-\frac{\kappa_1}{\theta_2^k} \left ( \gamma_1 (t_{\rm pre 1+})-\gamma_1(0)\right )}\frac{V^k\left ( 0\right )}{\theta_2^k}}, \forall t>t_{\rm pre 1},\\ 
            \lim_{t\rightarrow +\infty}\left \|\bar{\boldsymbol{x}}^k(t)-\bar{\boldsymbol{x}}^{k*} \right \|=0,
        \end{array}\right.
    \end{equation}
    where $\kappa_1 = \min \left \{m^k_{\min}-3\varsigma_{\rm max}^k,\frac{\lambda_2(L)(1- \beta_{\max}^k)}{2\psi_y^k} - 3\varsigma_{\rm max}^k\right.$, $\left. \frac{\varsigma_{\rm max}^k}{2}, \frac{\psi_d^k}{2}  \right \}$, $\psi_d^k = \min_{i\in\mathcal{V}}\left \{ \phi_i^k-\frac{1-\delta_i^k}{\alpha_i^k}\right \} $, $\psi_y^k = \max \left \{2 +  \frac{\lambda_N(L)}{\min_{i\in\mathcal{V}}\left \{\frac{1}{2\varsigma_{\rm max}^k}+l_{ii} \right \}},\frac{2\lambda_N(L)(1-\beta^k_{\max})}{\psi_d^k\min_{i\in\mathcal{V}}\left \{\alpha_i^k(\frac{1}{2\varsigma_{\rm max}^k}+l_{ii}) \right \}} \right \}$, $\theta_2^k = \max \left \{1,\frac{1}{2}+4\varsigma_{\rm max}^k, \frac{1}{2\lambda_2(L)}+4\varsigma_{\rm max}^k  \right \}$, $\varsigma_{\rm max}^k = \max_{i \in \mathcal{V}}\varsigma_i^k$ with $0<\varsigma_i^k < \min\left \{ \frac{m^k_{\min}}{3}, \frac{\lambda_2(L)(1- \beta_{\max}^k)}{6\psi_y^k}\right \}$, $m^k_{\min} = \min_{i\in\mathcal{V}}{m_i^k}$ and $\beta_{\max}^k = \max_{i \in \mathcal{V}} \beta_i^k$. $V^k(0)$ is the initial value of the Lyapunov function defined in (\ref{LyapunovV1}). Moreover, the weighting parameters $\omega_i$ also converge to $\omega_i^*$ within $t_{\rm pre1}$.
\end{lemma}

\begin{IEEEproof}
    Inspired by \cite{GUO2022distributeddynamic,yi2016initialization}, from Lemma \ref{propertiesDifferentialprojection}, we have 
    \begin{subequations}\label{Dealwithprojection}
        \begin{align}
            & \quad P_{T_{\boldsymbol{\Omega}}(\bar{\boldsymbol{x}}^k)}\left (  T_1(t,t_{\rm pre 1})\left ( \boldsymbol{y}^k-\nabla f^k\left (\bar{\boldsymbol{x}}^k\right )\right )\right )\notag\\ & = 
             T_1(t,t_{\rm pre 1}) \left ( \boldsymbol{y}^k-\nabla f^k\left (\bar{\boldsymbol{x}}^k\right )\right ) - C_{\boldsymbol{\Omega}}(\bar{\boldsymbol{x}}^k),\\
            & \quad P_{T_{\boldsymbol{\Omega}}(\bar{\boldsymbol{x}}^{k*})} \left ( T_1(t,t_{\rm pre 1})\left ( \boldsymbol{y}^{k*}-\nabla f^k\left (\bar{\boldsymbol{x}}^{k*}\right )\right )\right )\notag\\  & = 
           T_1(t,t_{\rm pre 1}) \left ( \boldsymbol{y}^{k*}-\nabla f^k\left (\bar{\boldsymbol{x}}^{k*}\right )\right ) - C_{\boldsymbol{\Omega}}(\bar{\boldsymbol{x}}^{k*}),
        \end{align}
    \end{subequations}
where $C_{\boldsymbol{\Omega}}(\bar{\boldsymbol{x}}^k)=\left [ \rho(\bar{x}_1^k)c_{\Omega_1}(\bar{x}_1^k),\dots,\rho(\bar{x}_N^k)c_{\Omega_N}(\bar{x}_N^k)\right ]^\top $ with $\rho(\bar{x}_i^k)\geq 0$, $c_{\Omega_i}(\bar{x}_i^k)\in n_{\Omega_i}(\bar{x}_i^k), i\in \mathcal{V}$.

Define $\Tilde{\boldsymbol{x}}^k = \bar{\boldsymbol{x}}^k-\bar{\boldsymbol{x}}^{k*}$, $\Tilde{\boldsymbol{y}}^k = \boldsymbol{y}^k-\boldsymbol{y}^{k*}$ and $\Tilde{\boldsymbol{z}}^k = \boldsymbol{z}^k-\boldsymbol{z}^{k*}$. From (\ref{solutionforsubproblemCompact}), (\ref{EQforsubproblemCompact}) and (\ref{Dealwithprojection}), we have
\begin{subequations}\label{ErrorforsubproblemCompact}
\begin{align}
        \dot{\Tilde {\boldsymbol{x}}}^{k} &= T_1(t,t_{\rm pre 1})\left (-\boldsymbol{h}^k+\Tilde{\boldsymbol{y}}^k \right )-C_{\boldsymbol{\Omega}}(\bar{\boldsymbol{x}}^k)+C_{\boldsymbol{\Omega}}(\bar{\boldsymbol{x}}^{k*}),\\
        \dot{\Tilde {\boldsymbol{y}}}^k & =   T_1(t,t_{\rm pre 1})\left ( - L\left ( \Tilde{\boldsymbol{y}}^k + \boldsymbol{e}_y^k\right )- \Tilde{\boldsymbol{z}}^k-\Tilde{\boldsymbol{x}}^k\right ),\\
        \dot{\Tilde {\boldsymbol{z}}}^{k} & =  T_1(t,t_{\rm pre 1})  L\left ( \Tilde{\boldsymbol{y}}^k + \boldsymbol{e}_y^k\right ),
\end{align} 
\end{subequations}
where $\boldsymbol{h}^k = \nabla f^k\left (\bar{\boldsymbol{x}}^k\right )-\nabla f^k\left (\bar{\boldsymbol{x}}^{k*}\right )$ and $\boldsymbol{e}_{y}^k = \left [e_{y,1}^k,\dots,e_{y,N}^k \right ]^\top $. 

Consider the following candidate Lyapunov function
\begin{equation}\label{LyapunovV1}
    V^k(t) = V^k_1(t) + V^k_2(t) + V^k_3(t),
\end{equation}
where $ V^k_1(t) = \frac{1}{2}\left ((\Tilde{\boldsymbol{x}}^k)^\top \Tilde{\boldsymbol{x}}^k+ (\Tilde{\boldsymbol{y}}^k)^\top \Tilde{\boldsymbol{y}}^k + (\Tilde{\boldsymbol{z}}^k)^\top  \varGamma \Tilde{\boldsymbol{z}}^k\right )$, $V^k_2(t) =   2\varsigma_{\rm max}^k\left (\Tilde{\boldsymbol{y}}^k+\Tilde{\boldsymbol{z}}^k \right )^\top \left (\Tilde{\boldsymbol{y}}^k+ \Tilde{\boldsymbol{z}}^k \right )$, $V^k_3(t)  = \sum_{i=1}^N \eta_i^k(t)$, $\varGamma  = Q\mathrm{diag}\left \{1,\Lambda^{-1} \right \}Q^\top $ defined in (\ref{orthogonaltransformation}), $\varsigma_{\rm max}^k = \max_{i \in \mathcal{V}}\varsigma_i^k$, and $\varsigma_i^k$ is a positive constant which will be defined later. 

Let $\boldsymbol{\varphi }^k = (\Tilde{\boldsymbol{x}}^k)^\top \Tilde{\boldsymbol{x}}^k + (\Tilde{\boldsymbol{y}}^k)^\top \Tilde{\boldsymbol{y}}^k + (\Tilde{\boldsymbol{z}}^k)^\top \Tilde{\boldsymbol{z}}^k + \sum_{i=1}^N \eta_i^k(t)$, we have  
\begin{equation}\label{subprobleminequality}
    \theta_1^k \boldsymbol{\varphi}^k \leq V^k(t) \leq \theta_2^k \boldsymbol{\varphi}^k,
\end{equation}
where $\theta_1^k = \min \left \{\frac{1}{2}, \frac{1}{2\lambda_N(L)} \right \}$ and $\theta_2^k = \max \left \{1,\frac{1}{2}+4\varsigma_{\rm max}^k, \frac{1}{2\lambda_2(L)}+4\varsigma_{\rm max}^k  \right \}$. 

Taking the derivative of $V^k_1(t)$, we have 
\begin{equation}\label{deV1}
    \begin{aligned}
    \dot{V}^k_1 =& T_1(t,t_{\rm pre 1})\left [-(\Tilde{\boldsymbol{x}}^k)^\top \boldsymbol{h}^k- (\Tilde{\boldsymbol{y}}^k)^\top  L \Tilde{\boldsymbol{y}}^k- (\Tilde{\boldsymbol{y}}^k)^\top   \Tilde{\boldsymbol{z}}^k\right. \\& \left. + (\Tilde{\boldsymbol{z}}^k)^\top  \varGamma L\Tilde{\boldsymbol{y}}^k - (\Tilde{\boldsymbol{y}}^k)^\top L\boldsymbol{e}_y^k + (\Tilde{\boldsymbol{z}}^k)^\top \varGamma L\boldsymbol{e}_y^k \right ] \\& + (\Tilde{\boldsymbol{x}}^k)^\top \left [ -C_{\boldsymbol{\Omega}}(\bar{\boldsymbol{x}}^k)+C_{\boldsymbol{\Omega}}(\bar{\boldsymbol{x}}^{k*})\right ].
    \end{aligned}
\end{equation}

From the definition of $n_{\Omega_i}(\bar{x}_i^k)$, we get $(\bar{x}_i^{k*}-\bar{x}_i^k)^\top  n_{\Omega_i}(\bar{x}_i^k)\leq 0$ and $(\bar{x}_i^{k}-\bar{x}_i^{k*})^\top  n_{\Omega_i}(\bar{x}_i^{k*})\leq 0$. Since $\rho(\bar{x}_i^k)c_{\Omega_N}(\bar{x}_i^k)\geq 0$, $\rho(\bar{x}_i^k)c_{\Omega_N}(\bar{x}_i^{k*})\geq 0$, $c_{\Omega_i}(\bar{x}_i^k)\in n_{\Omega_i}(\bar{x}_i^k)$ and $c_{\Omega_i}(\bar{x}_i^{k*})\in n_{\Omega_i}(\bar{x}_i^{k*})$, we have $(\bar{x}_i^{k*}-\bar{x}_i^k)^\top  c_{\Omega_i}(\bar{x}_i^k)\leq 0$ and  $(\bar{x}_i^{k}-\bar{x}_i^{k*})^\top  c_{\Omega_i}(\bar{x}_i^{k*})\leq 0$. Thus, based on the definitions of $\Tilde{\boldsymbol{x}}^k$ and $C_{\boldsymbol{\Omega}}(\bar{\boldsymbol{x}}^k)$, we obtain that the last term in (\ref{deV1}) satisfies 
\begin{equation}\label{substitute11}
    (\Tilde{\boldsymbol{x}}^k)^\top \left [ -C_{\boldsymbol{\Omega}}(\bar{\boldsymbol{x}}^k)+C_{\boldsymbol{\Omega}}(\bar{\boldsymbol{x}}^{k*})\right ] \leq 0.
\end{equation}

Based on Assumption \ref{assumptionMultiConvex} and Lemma \ref{lemmaStronglyConvex}, for the first term of (\ref{deV1}), we have 
\begin{equation}\label{34}
    \begin{aligned}
    -(\Tilde{\boldsymbol{x}}^k)^\top \boldsymbol{h}^k &= -(\Tilde{\boldsymbol{x}}^k)^\top \left [ \nabla f^k\left (\bar{\boldsymbol{x}}^k\right )-\nabla f^k\left (\bar{\boldsymbol{x}}^{k*}\right ) \right ]\\ & \leq -m^k_{\min} (\Tilde{\boldsymbol{x}}^k)^\top \Tilde{\boldsymbol{x}}^k.
    \end{aligned}
\end{equation}
where $m^k_{\min} = \min_{i\in\mathcal{V}}{m_i^k}$. According to (\ref{EQforsubproblemCompact}c), we have
\begin{equation}
\begin{aligned}
    &\quad -(\Tilde{\boldsymbol{y}}^k)^\top  L \left (\Tilde{\boldsymbol{y}}^k + \boldsymbol{e}_y^k\right ) \\ 
    &= - (\boldsymbol{y}^k)^\top  L \bar{\boldsymbol{y}}^k = -\left(\bar{\boldsymbol{y}}^k - \boldsymbol{e}_y^k\right )^\top L \bar{\boldsymbol{y}}^k\\
    & \stackrel{*}{=} -\sum_{i=1}^N \bar{q}_i^k(t) + \sum_{i=1}^N\sum_{j=1}^N a_{ij} e_{y,i}^k\left ( \bar{y}_i^k(t)-\bar{y}_j^k(t)\right )\\
    &\leq -\sum_{i=1}^N \bar{q}_i^k(t) + \sum_{i=1}^N\sum_{j=1}^N a_{ij} \left ( \left \| e_{y,i}^k \right \|^2 + \frac{1}{4}\left \| \bar{y}_i^k(t)-\bar{y}_j^k(t)\right \|^2\right )\\
    & \stackrel{**}{=} -\frac{1}{2}\sum_{i=1}^N \bar{q}_i^k(t) + \sum_{i=1}^N l_{ii} \left \| e_{y,i}^k \right \|^2,
\end{aligned}
\end{equation}
where the equalities $\stackrel{*}{=}$ and $\stackrel{**}{=}$ hold due to 
\begin{equation}\label{sumq}
    \sum_{i=1}^N \bar{q}_i^k(t) = \frac{1}{2} \sum_{i=1}^N\sum_{j=1}^N a_{ij}\left \| \bar{y}_i^k(t)-\bar{y}_j^k(t)\right \|^2 = \left(\bar{\boldsymbol{y}}^k\right )^\top L \bar{\boldsymbol{y}}^k.
\end{equation}

In terms of the term $-(\Tilde{\boldsymbol{y}}^k)^\top   \Tilde{\boldsymbol{z}}^k$, we have $q_1^\top  \Tilde{\boldsymbol{z}}^k = 0$ due to $\sum_{i=1}^{N} z_i^k(t)=0$. Thus, according to Lemma \ref{orthogonaltransformation},
\begin{equation}
    -(\Tilde{\boldsymbol{y}}^k)^\top   \Tilde{\boldsymbol{z}}^k = - (\Tilde{\boldsymbol{y}}^k)^\top (QQ^\top )\Tilde{\boldsymbol{z}}^k = -(\Tilde{\boldsymbol{y}}^k)^\top (Q_2Q_2^\top )\Tilde{\boldsymbol{z}}^k,
\end{equation}
\begin{equation}
    (\Tilde{\boldsymbol{z}}^k)^\top \varGamma L\boldsymbol{e}_y^k = (\Tilde{\boldsymbol{z}}^k)^\top \left ( I - \frac{1}{N}\boldsymbol{1}_N\boldsymbol{1}_N^\top \right )\boldsymbol{e}_y^k 
    = (\Tilde{\boldsymbol{z}}^k)^\top \boldsymbol{e}_y^k.  
\end{equation}

The term $(\Tilde{\boldsymbol{z}}^k)^\top  \varGamma L\Tilde{\boldsymbol{y}}^k$ can be rewritten as 
\begin{equation}\label{substitute15}
\begin{aligned}
        (\Tilde{\boldsymbol{z}}^k)^\top  \varGamma L\Tilde{\boldsymbol{y}}^k
        & = (\Tilde{\boldsymbol{z}}^k)^\top  Q\begin{bmatrix}
            1& \\ 
            &\Lambda^{-1} 
            \end{bmatrix} Q^\top Q\begin{bmatrix}
            0& \\ 
            &\Lambda 
            \end{bmatrix} Q^\top \Tilde{\boldsymbol{y}}^k\\
        &= (\Tilde{\boldsymbol{z}}^k)^\top \left ( Q\begin{bmatrix}
 0& \\ 
 & I_{N-1}
\end{bmatrix}Q^\top \right )\Tilde{\boldsymbol{y}}^k\\
&= (\Tilde{\boldsymbol{z}}^k)^\top  (Q_2Q_2^\top )\Tilde{\boldsymbol{y}}^k.
\end{aligned}
\end{equation}

Using Young's inequality, we have
\begin{equation}\label{V1kyoung}
    (\Tilde{\boldsymbol{z}}^k)^\top \boldsymbol{e}_y^k \leq \frac{\varsigma_{\rm max}^k}{2}(\Tilde{\boldsymbol{z}}^k)^\top \Tilde{\boldsymbol{z}}^k +  \frac{1}{2\varsigma_{\rm max}^k}(\boldsymbol{e}_y^k)^\top \boldsymbol{e}_y^k.
\end{equation}

Substituting (\ref{substitute11})–(\ref{V1kyoung}) into (\ref{deV1}), one can obtain
\begin{equation}\label{deV1leq} 
    \begin{aligned}
    \dot{V}_1^k \leq &  T_1(t,t_{\rm pre 1})\left [-m^k_{\min} (\Tilde{\boldsymbol{x}}^k)^\top \Tilde{\boldsymbol{x}}^k  + \frac{\varsigma_{\rm max}^k}{2}(\Tilde{\boldsymbol{z}}^k)^\top \Tilde{\boldsymbol{z}}^k\right.\\  &\left. -\frac{1}{2}\sum_{i=1}^N \bar{q}_i^k(t) + \sum_{i=1}^N\left ( \frac{1}{2\varsigma_{\rm max}^k} + l_{ii}\right ) \left \| e_{y,i}^k \right \|^2\right ].
\end{aligned}
\end{equation}

The time derivative of $V_2^k$ is 
\begin{equation}
    \dot{V}_2^k  =   2\varsigma_{\rm max}^k T_1(t,t_{\rm pre 1})\left (\Tilde{\boldsymbol{y}}^k+\Tilde{\boldsymbol{z}}^k \right )^\top \left (-\Tilde{\boldsymbol{z}}^k -\Tilde{\boldsymbol{x}}^k \right ).
\end{equation}

According to Young's inequality, we have $ - (\Tilde{\boldsymbol{y}}^k)^\top \Tilde{\boldsymbol{z}}^k \leq \frac{1}{4} (\Tilde{\boldsymbol{z}}^k)^\top \Tilde{\boldsymbol{z}}^k+(\Tilde{\boldsymbol{y}}^k)^\top \Tilde{\boldsymbol{y}}^k$, $- (\Tilde{\boldsymbol{z}}^k)^\top \Tilde{\boldsymbol{x}}^k \leq \frac{1}{4} (\Tilde{\boldsymbol{z}}^k)^\top \Tilde{\boldsymbol{z}}^k+(\Tilde{\boldsymbol{x}}^k)^\top \Tilde{\boldsymbol{x}}^k$, $- (\Tilde{\boldsymbol{y}}^k)^\top \Tilde{\boldsymbol{x}}^k \leq \frac{1}{2} (\Tilde{\boldsymbol{y}}^k)^\top \Tilde{\boldsymbol{y}}^k+ \frac{1}{2} (\Tilde{\boldsymbol{x}}^k)^\top \Tilde{\boldsymbol{x}}^k$. Then we have 
\begin{equation}\label{deV2leq}
    \dot{V}_2^k  \leq   \varsigma_{\rm max}^k T_1(t,t_{\rm pre 1})\left [3(\Tilde{\boldsymbol{x}}^k)^\top \Tilde{\boldsymbol{x}}^k+3(\Tilde{\boldsymbol{y}}^k)^\top \Tilde{\boldsymbol{y}}^k-(\Tilde{\boldsymbol{z}}^k)^\top \Tilde{\boldsymbol{z}}^k \right ].
\end{equation}

Calculating the time derivative of $V^k_3(t)$ yields
\begin{equation}\label{deV3leq}
\begin{aligned}
    \dot{V}^k_3(t) 
    & \leq T_1(t,t_{\rm pre 1})\sum_{i=1}^N \left (-\phi_i^k\eta_i^k(t) \right.\\ &\quad \left.- \delta_i^k ( \frac{1}{2\varsigma_{\rm max}^k}+l_{ii}) \left \|e_{y,i}^k(t) \right \|^2 + \delta_i^k\frac{\beta_i^k}{2}\bar{q}_i^k(t)\right ).
\end{aligned}  
\end{equation}

Thus, adding the (\ref{deV1leq}), (\ref{deV2leq}) and (\ref{deV3leq}) yields
\begin{equation}\label{dVmid}
\begin{aligned}
    \dot{V}^k \leq &  T_1(t,t_{\rm pre 1})\left [ -(m^k_{\min}-3\varsigma_{\rm max}^k)(\Tilde{\boldsymbol{x}}^k)^\top \Tilde{\boldsymbol{x}}^k  \right.\\ 
        & \left.+ 3\varsigma_{\rm max}^k (\Tilde{\boldsymbol{y}}^k)^\top  \Tilde{\boldsymbol{y}}^k  - \frac{\varsigma_{\rm max}^k}{2}(\Tilde{\boldsymbol{z}}^k)^\top  \Tilde{\boldsymbol{z}}^k - \sum_{i=1}^N \phi_i^k\eta_i^k(t) \right.\\ 
        & \left.+ \sum_{i=1}^N (1-\delta_i^k) (\frac{1}{2\varsigma_{\rm max}^k}+l_{ii})\left \|e_{y,i}^k \right \|^2 \right.\\ 
        & \left.+ \frac{1}{2}\sum_{i=1}^N \left ( \delta_i^k \beta_i^k - 1\right )\bar{q}_i^k(t) \right ].
\end{aligned}
\end{equation}

Let $\varDelta  = - \sum_{i=1}^N \phi_i^k\eta_i^k(t) + \sum_{i=1}^N (1-\delta_i^k) (\frac{1}{2\varsigma_{\rm max}^k}+l_{ii})\left \|e_{y,i}^k \right \|^2 + \frac{1}{2}\sum_{i=1}^N \left ( \delta_i^k \beta_i^k - 1\right )\bar{q}_i^k(t)$. From (\ref{thresholdETM1}) and (\ref{sumq}), defining $\beta_{\max}^k = \max_{i \in \mathcal{V}} \beta_i^k $, we have 
\begin{equation}
    \varDelta 
    \leq -\sum_{i=1}^N \left ( \phi_i^k-\frac{1-\delta_i^k}{\alpha_i^k}\right )\eta_i^k(t) -\frac{1}{2}\left (1- \beta_{\max}^k)\right (\boldsymbol{\bar{y}}^k)^\top L\boldsymbol{\bar{y}}^k.
\end{equation}

From equation (20) in \cite{duDOdynamic}, we get 
\begin{equation}
    \begin{aligned}
        (\boldsymbol{y}^k)^\top L \boldsymbol{y}^k & \leq \psi_y^k (\boldsymbol{\bar{y}}^k)^\top L(\boldsymbol{\bar{y}}^k ) \\
        &\quad + \frac{2\lambda_N(L)}{\min_{i\in\mathcal{V}}\left \{\alpha_i^k(\frac{1}{2\varsigma_{\rm max}^k}+l_{ii}) \right \}}\sum_{i=1}^N \eta_i^k(t),
    \end{aligned}
\end{equation}
where $\psi_d^k = \min_{i\in\mathcal{V}}\left \{ \phi_i^k-\frac{1-\delta_i^k}{\alpha_i^k}\right \} $ and $\psi_y^k = \max \left \{2 +  \frac{\lambda_N(L)}{\min_{i\in\mathcal{V}}\left \{\frac{1}{2\varsigma_{\rm max}^k}+l_{ii} \right \}},\frac{2\lambda_N(L)(1-\beta^k_{\max})}{\psi_d^k\min_{i\in\mathcal{V}}\left \{\alpha_i^k(\frac{1}{2\varsigma_{\rm max}^k}+l_{ii}) \right \}} \right \}$.

Thus, together with (\ref{thresholdETM1}) and (\ref{EQforsubproblemCompact}c), we have $-\frac{1}{2}\left (1- \beta_{\max}^k)\right (\boldsymbol{\bar{y}}^k)^\top L\boldsymbol{\bar{y}}^k \leq -\frac{1}{2\psi_y^k}\left (1- \beta_{\max}^k)\right (\boldsymbol{\Tilde{y}}^k)^\top L \boldsymbol{\Tilde{y}}^k + \frac{\psi_d^k}{2} \sum_{i=1}^N \eta_i^k(t)$, which implies that 
\begin{equation}\label{deltafinal}
    \varDelta \leq - \frac{\psi_d^k}{2} \sum_{i=1}^N \eta_i^k(t)-\frac{1}{2\psi_y^k}\left (1- \beta_{\max}^k)\right (\boldsymbol{\Tilde{y}}^k)^\top L \boldsymbol{\Tilde{y}}^k.
\end{equation}

Substituting (\ref{deltafinal}) into (\ref{dVmid}) leads to 
\begin{equation}\label{ref66}
\begin{aligned}
    \dot{V}^k \leq &  T_1(t,t_{\rm pre 1})\left [ -(m^k_{\min}-3\varsigma_{\rm max}^k)(\Tilde{\boldsymbol{x}}^k)^\top \Tilde{\boldsymbol{x}}^k  \right.\\ 
        & \left.- \left (\frac{\lambda_2(L)(1- \beta_{\max}^k)}{2\psi_y^k} - 3\varsigma_{\rm max}^k\right )  (\Tilde{\boldsymbol{y}}^k)^\top  \Tilde{\boldsymbol{y}}^k \right.\\ 
        & \left. -\frac{\varsigma_{\rm max}^k}{2}(\Tilde{\boldsymbol{z}}^k)^\top  \Tilde{\boldsymbol{z}}^k - \frac{\psi_d^k}{2} \sum_{i=1}^N \eta_i^k(t) \right ].
\end{aligned}
\end{equation}

Let $0< \varsigma_i^k < \min\left \{ \frac{m^k_{\min}}{3}, \frac{\lambda_2(L)(1- \beta_{\max}^k)}{6\psi_y^k}\right \}$ for all $i \in \mathcal{V}$, together with (\ref{subprobleminequality}) we have 
\begin{equation}\label{deVleqETM}
    \dot{V}^k \leq - \frac{\kappa_1}{\theta_2^k}T_1(t,t_{\rm pre 1})V^k,
\end{equation}
where $\kappa_1 = \min \left \{m^k_{\min}-3\varsigma_{\rm max}^k,\frac{\lambda_2(L)(1- \beta_{\max}^k)}{2\psi_y^k} - 3\varsigma_{\rm max}^k\right.$, $\left. \frac{\varsigma_{\rm max}^k}{2}, \frac{\psi_d^k}{2}  \right \}$.

Recalling Lemma \ref{prescribedcondition} and (\ref{subprobleminequality}), with $T_1(t,t_{\rm pre1})$ defined in (\ref{definitionTBG}), the algorithm (\ref{solutionforsubproblem}) achieves prescribed-time approximate convergence at $t_{\rm pre1}$. In addition, the error is $\epsilon_1=\sqrt{e^{-\frac{\kappa_1}{\theta_2^k} \left ( \gamma_1 (t_{\rm pre 1+})-\gamma_1(0)\right )}\frac{V^k\left ( 0\right )}{\theta_2^k}}$.

Next, we prove that the Zeno behavior is excluded based on the contradiction method. For $k$th objective, assume the Zeno behavior occurs at $T_0^k$, namely, $\lim_{\ell \rightarrow \infty}t_{i,k}^{\ell} = T_0^k>0, \exists i\in \mathcal{V}$. From (\ref{deVleqETM}), we obtain that there exists a positive upper bound $M_0^k>0$ such that $\|\dot{e}_{y,i}^k(t)\|=\|\dot{y}_i^k(t)\| \leq M_0^k, \forall t\geq 0$. Let $\varepsilon_0^k = \frac{1}{2M_0^k}\sqrt{\frac{\eta_i^k(0)}{\alpha_i^k(\frac{1}{2\varsigma_i^k}+l_{ii})}} e^{ -\frac{1}{2}(\phi_i^k + \delta_i^k/\alpha_i^k) \gamma(t,\sigma) } >0$, then there exists a positive integer $N(\varepsilon_0^k)$ such that 
\begin{equation}\label{Zeno1}
    t_{i,k}^{\ell} \in \left [ T_0^k-\varepsilon_0^k,T_0^k \right ],\forall \ell>N(\varepsilon_0^k).
\end{equation}
Noting that $\bar{q}_i^k>0$, together with (\ref{etageq0}), a necessary condition to guarantee (\ref{thresholdETM1}) is $\left \| e_{y,i}^k(t) \right \| \geq \sqrt{\frac{\eta_i^k(0)}{\alpha_i^k(\frac{1}{2\varsigma_i^k}+l_{ii})}} e^{ -\frac{1}{2}(\phi_i^k + \delta_i^k/\alpha_i^k) \gamma(t,\sigma) }>0$. Invoking $\|\dot{e}_{y,i}^k(t)\| \leq M_0^k, \forall t\geq 0$, together with $e_{y,i}^k\left(t_{i,k}^{N(\varepsilon_0^k)}\right)=0$, we have 
\begin{equation}
    \|e_{y,i}^k\left(t_{i,k}^{N(\varepsilon_0^k)+1}\right)\| \leq \left(t_{i,k}^{N(\varepsilon_0^k)+1}- t_{i,k}^{N(\varepsilon_0^k)}\right) M_0^k.
\end{equation}
Then, 
\begin{equation}
\begin{aligned}
    & \quad t_{i,k}^{N(\varepsilon_0^k)+1}- t_{i,k}^{N(\varepsilon_0^k)}\\
    & \geq \frac{1}{M_0}\sqrt{\frac{\eta_i^k(0)}{\alpha_i^k(\frac{1}{2\varsigma_i^k}+l_{ii})}} e^{ -\frac{1}{2}(\phi_i^k + \delta_i^k/\alpha_i^k) \gamma(t,\sigma) } = 2\varepsilon_0^k,
\end{aligned} 
\end{equation}
which contradicts to (\ref{Zeno1}). Therefore, the Zeno behavior is circumvented. The proof is completed.
\end{IEEEproof}

Then, the following lemma presents the convergence of the ideal point seeking module \eqref{idealpoint}.

\begin{lemma}\label{lemmaideal}
    Under Assumption \ref{assumptionMultiConvex}, with the algorithm (\ref{idealpoint}), $\hat{x}_i^k$ converges to the ideal solution $\hat{x}_i^{k*}$ at a prescribed time $t_{\rm pre 2}$ for all $k\in \mathcal{K},i\in\mathcal{V}$, and the error is bounded by
    \begin{equation}\label{CE2}
        \left\{\begin{array}{@{}l}
            \lim_{t\rightarrow t_{\rm pre 2+}}\left \|\hat{\boldsymbol{x}}^k(t)-\hat{\boldsymbol{x}}^{k*} \right \|\\
            \quad \quad \quad \quad \quad \quad \quad \leq \sqrt{e^{-2m^k_{\min} \left ( \gamma_2 (t_{\rm pre 2+})-\gamma_2(0)\right )}2\hat{V}^k\left ( 0\right )},\\ 
            \left \|\hat{\boldsymbol{x}}^k(t)-\hat{\boldsymbol{x}}^{k*} \right \|\leq \sqrt{e^{-2m^k_{\min} \left ( \gamma_2 (t_{\rm pre 2+})-\gamma_2(0)\right )}2\hat{V}^k\left ( 0\right )},\\ 
            \lim_{t\rightarrow \infty}\left \|\hat{\boldsymbol{x}}^k(t)-\hat{\boldsymbol{x}}^{k*} \right \|=0,
        \end{array}\right.
    \end{equation}
     where $\hat{V}^k(0)$ is the initial value of the Lyapunov function $\hat{V}^k(t) = \frac{1}{2}(\check{\boldsymbol{x}}^k)^\top \check{\boldsymbol{x}}^k$.
\end{lemma}
\begin{IEEEproof}
    Define $\check{\boldsymbol{x}}^k = \hat{\boldsymbol{x}}^k-\hat{\boldsymbol{x}}^{k*}$ and the Lyapunov candidate function $ \hat{V}^k(t) = \frac{1}{2}(\check{\boldsymbol{x}}^k)^\top \check{\boldsymbol{x}}^k$. Taking the time derivative, from (\ref{Dealwithprojection})-(\ref{34}), we have:
    \begin{equation}
    \begin{aligned}
        \dot{\hat{V}}^k &= T_2(t,t_{\rm pre 2})(\check{\boldsymbol{x}}^k)^\top \left (\hat{\boldsymbol{h}}^k -C_{\boldsymbol{\Omega}}(\hat{\boldsymbol{x}}^k)+C_{\boldsymbol{\Omega}}(\hat{\boldsymbol{x}}^{k*}) \right)\\
        &\leq 
        -2m^k_{\min}T_2(t,t_{\rm pre 2})\hat{V}^k(t),
    \end{aligned}   
    \end{equation}
    where $\hat{\boldsymbol{h}}^k = \nabla f^k\left (\hat{\boldsymbol{x}}^k\right )-\nabla f^k\left (\hat{\boldsymbol{x}}^{k*}\right )$. Therefore, $\epsilon_2=\sqrt{e^{-2m^k_{\min} \left ( \gamma_2 (t_{\rm pre 2+})-\gamma_2(0)\right )}2\hat{V}^k\left ( 0\right )}$
\end{IEEEproof}

Finally, letting $\boldsymbol{\nu} = \left [ \nu_1,\dots,\nu_N\right ]^\top \in \mathbb{R}^N$ and
$\boldsymbol{\mu} = \left [ \mu_1,\dots,\mu_N\right ]^\top \in \mathbb{R}^N$, the following theorem demonstrates the overall convergence of the proposed algorithm (\ref{solutionforsubproblem})-(\ref{dynamicETM2}) for the DCMRAP (\ref{problemMultiobjectiveresourceallocation}).

\begin{lemma}[Optimality]\label{lemmaOptimality}
    Under Assumptions \ref{MultiAssumptionGraph}, \ref{assumptionMultiSlater} and \ref{assumptionconvexu}, for any bounded initial points $x_i(0)\in \Omega_i, \forall i \in \mathcal{V}$, if $\left (\boldsymbol{x}^{*}, \boldsymbol{\nu}^{*}, \boldsymbol{\mu}^{*}\right )$ is the equilibrium of (\ref{solutionforproblem})
    , then $\boldsymbol{x}^{*}$ is the optimal solution of the problem (\ref{problemREMultiobjectiveresourceallocation}).
\end{lemma}
\begin{IEEEproof}
The proof is similar to the proof of Lemma \ref{lemmaSubOptimality} and is therefore omitted.
\end{IEEEproof}
\begin{theorem}\label{theoremtpre2}
    Under Assumptions \ref{MultiAssumptionGraph}, \ref{assumptionMultiSlater} and \ref{assumptionconvexu}, with Lemmas \ref{lemmasub} and \ref{lemmaideal}, the proposed algorithm in (\ref{solutionforproblem}) with the dynamic ETM (\ref{thresholdETM2})-(\ref{dynamicETM2}) solves the DCMRAP (\ref{problemMultiobjectiveresourceallocation}) in a prescribed time $t_{\rm pre3}>\max\left \{t_{\rm pre 1},t_{\rm pre 2}\right\}$, and the Zeno behavior is excluded, and the convergence error is bounded by
    \begin{equation}\label{CE3}
        \left\{\begin{array}{@{}l}
            \lim_{t\rightarrow t_{\rm pre 3+}}\left \|\boldsymbol{x}(t)-\boldsymbol{x}^{*} \right \|\leq \sqrt{e^{-\frac{\kappa_2}{\vartheta_2} \left ( \gamma_3 (t_{\rm pre 3+})-\gamma_3(0)\right )}\frac{V\left ( 0\right )}{\vartheta_2^k}},\\ 
            \left \|\boldsymbol{x}(t)-\boldsymbol{x}^{*} \right \|\leq \sqrt{e^{-\frac{\kappa_2}{\vartheta_2} \left ( \gamma_3 (t_{\rm pre 3+})-\gamma_3(0)\right )}\frac{V\left ( 0\right )}{\vartheta_2}}, \forall t>t_{\rm pre 3},\\ 
            \lim_{t\rightarrow +\infty}\left \|\boldsymbol{x}(t)-\boldsymbol{x}^{*} \right \|=0,
        \end{array}\right.
    \end{equation}
    where $\kappa_2 = \min \left \{\varpi_{\min}-3\varsigma_{\rm max},\frac{\lambda_2(L)(1- \beta_{\max})}{2\psi_\nu} - 3\varsigma_{\rm max}\right.$, $\left. \frac{\varsigma_{\rm max}}{2}, \frac{\psi_d}{2}\right \}$, $\psi_d = \min_{i\in\mathcal{V}}\left \{ \phi_i-\frac{1-\delta_i}{\alpha_i}\right \} $, $\vartheta_2 = \max \left \{1,\frac{1}{2}+4\varsigma_{\rm max}, \frac{1}{2\lambda_2(L)}+4\varsigma_{\rm max}  \right \}$, $\psi_\nu = \max \left \{2 +  \frac{\lambda_N(L)}{\min_{i\in\mathcal{V}}\left \{\frac{1}{2\varsigma_{\rm max}}+l_{ii} \right \}},  \frac{2\lambda_N(L)(1-\beta_{\max})}{\psi_d\min_{i\in\mathcal{V}}\left \{\alpha_i(\frac{1}{2\varsigma_{\rm max}}+l_{ii}) \right \}} \right \}$, $\varsigma_{\rm max} = \max_{i \in \mathcal{V}} \varsigma_i$ with $\varsigma_i < \min\left \{ \frac{\varpi_{\min}}{3}, \frac{\lambda_2(L)(1- \beta_{\max})}{6\psi_\nu}\right \}$, $\varpi_{\min} = \min_{i\in\mathcal{V}}{\varpi_i}$ and $\beta_{\max} = \max_{i \in \mathcal{V}} \beta_i$. $V(0)$ is the initial value of the Lyapunov function defined in (\ref{LyapunovRE}).  
\end{theorem}
\begin{IEEEproof}
    To prove Theorem \ref{theoremtpre2}, note that the initial time of the TBG $T_3(t,t_{\rm pre 3})$ is zero, hence the key is to prove that the convergence of (\ref{solutionforproblem}) for the problem (\ref{problemREMultiobjectiveresourceallocation}) is achieved at $t_{\rm pre 3}$ when $t>\max\left \{t_{\rm pre 1},t_{\rm pre 2}\right\}$. Define $\Tilde{\boldsymbol{x}} = \boldsymbol{x}-\boldsymbol{x}^*$, $\Tilde{\boldsymbol{\mu}} = \boldsymbol{\mu}-\boldsymbol{\mu}^*$ and $\Tilde{\boldsymbol{\nu}} = \boldsymbol{\nu}-\boldsymbol{\nu}^*$. Consider the following Lyapunov function
\begin{equation}\label{LyapunovRE}
   V(t) = V_1(t) + V_2(t) + V_3(t),
\end{equation}
where $V_1(t) = \frac{1}{2}\left ((\Tilde{\boldsymbol{x}})^\top \Tilde{\boldsymbol{x}}+ (\Tilde{\boldsymbol{\nu}})^\top \Tilde{\boldsymbol{\nu}}+(\Tilde{\boldsymbol{\mu}})^\top  \varGamma \Tilde{\boldsymbol{\mu}}\right ), V_2(t)  =   2\varsigma_i\left (\Tilde{\boldsymbol{\nu}}+\Tilde{\boldsymbol{\mu}} \right )^\top \left (\Tilde{\boldsymbol{\nu}}+\Tilde{\boldsymbol{\mu}} \right ), V_3(t)  = \sum_{i=1}^N \eta_i(t)$, the convergence of the system can be proved similarly as in the proof of Lemma \ref{lemmasub}, so the detailed procedure is omitted.
\end{IEEEproof}

\begin{remark}
Under Assumptions \ref{assumptionMultiConvex} and \ref{assumptionconvexu}, strong convexity underlies inequalities \eqref{subprobleminequality} and \eqref{deVleqETM}, which Lemma \ref{prescribedcondition} uses to guarantee prescribed‐time convergence with explicit error bounds \eqref{CE1}, \eqref{CE2} and \eqref{CE3}. Without strong convexity, these bounds collapse and finite-time convergence cannot be ensured. Although the ETMs may remain operable—since the triggering error term in \eqref{deltafinal} stays non‐positive—their performance would be re‐established. Nonetheless, the TBGs still accelerate transient behavior, and asymptotic optimality may be achieved as in first‐order methods \cite{yi2016initialization}.
\end{remark}

For real-time applications, Algorithm \ref{alg:alg1} encapsulates the overarching structure of the proposed algorithm tailored for each agent. Three key steps are included in the process: initialization, preference formulation and compromise. During the initialization stage, all initial values for the algorithm parameters and the prescribed times are set. The preference index is then formulated online by executing \eqref{solutionforsubproblem}-\eqref{idealpoint}, which determines the optimal weightings and ideal points within a prescribed time. Finally, a compromised solution is obtained from (\ref{solutionforproblem})-(\ref{dynamicETM2}) based on the preference index. 

\begin{algorithm}[ht!]
\caption{Event-triggered Prescribed-Time Distributed Algorithm for Agent $i$}\label{alg:alg1}
\begin{algorithmic}
\STATE 
\STATE \textbf{Initialization:}
\STATE \hspace{0.3cm} For $i\in \mathcal{V}$, $k\in \mathcal{K}$, set
\STATE \hspace{0.7cm} $\bar{x}_i^k = \bar{x}_i^k(0), y_i^k=y_i^k(0), z_i^k = z_i^k(0)=0,$
\STATE \hspace{0.7cm} $\omega_i^k = \omega_i^k(0),\eta_i^k= \eta_i^k(0), x_i = x_i(0),$ 
\STATE \hspace{0.7cm} $\nu_i = \nu_i(0), \mu_i=\mu_i(0)=0,\eta_i= \eta_i(0).$
\STATE \hspace{0.3cm} choose $\phi_i^k, \delta_i^k, \alpha_i^k,\beta_i^k,\varsigma_i^k, \phi_i, \delta_i, \alpha_i ,\beta_i,\varsigma_i$, 
\STATE \hspace{0.3cm} choose the prescribed settling times, $t_{\rm pre 1}$, $t_{\rm pre 2}$ 
\STATE \hspace{0.7cm} and $t_{\rm pre 3}$, $ \max\left \{t_{\rm pre 1},t_{\rm pre 2} \right \}< t_{\rm pre 3}$.
\STATE \textbf{Preference formulation:}
\STATE \hspace{0.5cm} Run (\ref{solutionforsubproblem})-(\ref{idealpoint}).
\STATE \textbf{Compromise:}
\STATE \hspace{0.5cm} Run (\ref{solutionforproblem})-(\ref{dynamicETM2}).
\STATE \textbf{End if} convergence to the optimal solution is achieved.
\end{algorithmic}
\label{alg1}
\end{algorithm}

\begin{remark}
   Unlike previous studies on distributed multiobjective optimization \cite{li2021multiobjective, liu2023multiobjective}, which only address global equality constraints, the proposed algorithm tackles CMRAPs with both global and local constraints. A tailored projection operator, $ P_{T_{\Omega_i}(x_i)} $, is introduced to handle local set constraints, integrating a generalized TBG to ensure prescribed-time convergence. Furthermore, novel formulations are provided for weighting coefficients (\ref{problemforkthobjective}) and ideal points (\ref{idealpoint}), accounting for the feasible decision variable set. Solving problem \eqref{problemREMultiobjectiveresourceallocation} produces a global optimal solution, but this solution may not be network‐level Pareto optimal in general.

\end{remark}

\begin{remark}\label{remark4}
The classical TBG in \cite{NING2019Practicalfixed, GUO2022distributeddynamic, yang2025preETM} is a specific case of the generalized TBGs proposed in \cite{liu2023multiobjective}. Generalized TBGs provide greater design flexibility by accommodating a wider range of functional formulas, which allows our algorithms to be tailored for improved solution accuracy or smoother convergence process, depending on the application requirements.
\end{remark}
\begin{remark}\label{remark5}
    Novel dynamic ETMs are introduced to enhance the practical applicability of our approach and reduce communication overhead. By incorporating the TBG $T_1(t,t_{\rm pre 1})$ (respectively, $T_3(t,t_{\rm pre 3})$) into the dynamic evolution law (\ref{dynamicETM}) (respectively, (\ref{dynamicETM2})), the proposed ETMs meet the stringent requirements of prescribed-time convergence. Designing and tuning the generalized TBGs to modify the dynamics of the internal variables $\eta_i(t)$ and $\eta_i^k(t)$ adjusts the performance of the ETMs. This strategy not only reduces unnecessary data transmission but also provides more flexibility for the design of communication performance without affecting the convergence speed and accuracy.
\end{remark}
\begin{remark}\label{remark6}
    Given that $\eta_i^k(t) > 0$ for all $t \geq 0$, the Zeno behavior is excluded in the dynamic ETM (\ref{thresholdETM1})--(\ref{dynamicETM}). The condition $0<\varsigma_i^k < \min\left \{ \frac{m^k_{\min}}{3}, \frac{\lambda_2(L)(1- \beta_{\max}^k)}{6\psi_y^k}\right \}$ is introduced to prove the stability via Lyapunov functions, although this may present a conservative bound. Notably, for smaller values of $m^k_{\min}$, $\varsigma_i^k$ invariably remains less than $\frac{m^k_{\min}}{3}$, rendering the complex parameter $\frac{\lambda_2(L)(1- \beta_{\max}^k)}{6\psi_y^k}$ non-essential. Furthermore, adjusting the parameters $\phi_i^k, \delta_i^k, \alpha_i^k$ and $\beta_i^k$ can achieve a higher $\frac{\lambda_2(L)(1- \beta_{\max}^k)}{6\psi_y^k}$, ensuring that $\frac{m^k_{\min}}{3} < \frac{\lambda_2(L)(1- \beta_{\max}^k)}{6\psi_y^k}$. The same approach applies to parameter selection for the dynamic ETM (\ref{thresholdETM2})--(\ref{dynamicETM2}). 
\end{remark}
    
    \begin{remark}\label{remark7}
    In scenarios where the parameter $\alpha_i$ tends towards infinity within the dynamic triggering law (\ref{dynamicETM2}), our proposed dynamic ETM (\ref{thresholdETM2})--(\ref{dynamicETM2}) defaults to the static ETM in \cite{gao2022ETMGT, zhao2019distributed, wu2020event, chai2024event}, as described below:
    \begin{equation}\label{thresholdETMstatic}
    \begin{split}
        t_{i}^{\ell+1} = \inf_{t\geq t_{i}^\ell}\Bigg\{t \mid    &\left(\frac{1}{2\varsigma_i}+l_{ii}\right)\left \|e_{\nu,i}(t) \right \|^2  - \frac{\beta_i}{2} \bar{q}_i(t)  \geq 0  \Bigg\}.
    \end{split}
    \end{equation}
    By comparison, the threshold in (\ref{thresholdETM2}) exceeds that in (\ref{thresholdETMstatic}), resulting in reduced communication cost. 

    With $\beta_i=0$, the dynamic ETM (\ref{thresholdETM2})--(\ref{dynamicETM2}) aligns with the method presented in \cite{GUO2022distributeddynamic}:
    
    \begin{equation}\label{thresholdETMdynamicguo}
        t_{i}^{\ell+1} = \inf_{t\geq t_{i}^\ell}\left\{t \mid  \alpha_i \left(\frac{1}{2\varsigma_i}+l_{ii}\right)\left \|e_{\nu,i}(t) \right \|^2 \geq \mathcal{H}_i(t)  \right\},
    \end{equation}
    where $\mathcal{H}_i(t) = \eta_i(t)$. $\eta_i(t)$ is updated by
    \begin{equation}\label{dynamicETMguo}
        \dot{\eta}_i(t) = T_1(t,t_{\rm pre 3})\left ( -\phi_i\eta_i(t) -\delta_i  (\frac{1}{2\varsigma_i}+l_{ii})\left \|e_{\nu,i}(t) \right \|^2 \right ).
    \end{equation}
    The introduction of a network-based disagreement term $\bar{q}_i(t)$ in (\ref{thresholdETM2}) establishes a higher threshold $\mathcal{H}_i(t)$ compared to (\ref{thresholdETMdynamicguo}), which lacks this term. Moreover, (\ref{dynamicETM2}) includes a dynamic feedback term dependent on both local error metrics and network disagreement levels, thereby offering a dynamically adjustable threshold that accounts for both individual performance and interaction effects across the network, potentially leading to more robust system behavior.
    
     In conclusion, the static ETM (\ref{thresholdETMstatic}) and the dynamic ETM (\ref{thresholdETMdynamicguo})--(\ref{dynamicETMguo}) are specific cases of the proposed dynamic ETM. These attributes equally apply to the ETM (\ref{thresholdETM1})-(\ref{dynamicETM}).
\end{remark}

\section{SIMULATION EXPERIMENTS}\label{section4}
\begin{figure}[b]
    \centering
    \includegraphics[width=0.7\linewidth]{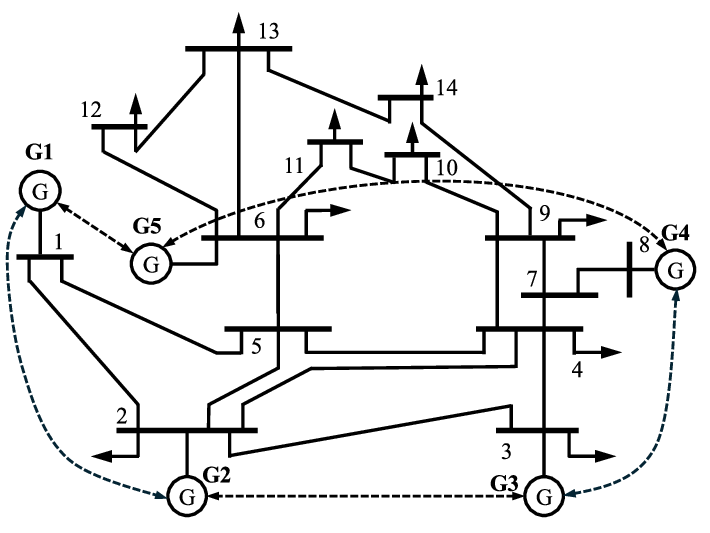}
    \caption{IEEE 14-bus system with communication network.}
    \label{fig:microgrids}
\end{figure}
\begin{table*}[ht!]
    \centering
    \caption{Parameters of the Generators.}
    \label{table:parameters}
    \begin{tabular}{ccccccccccccc}
      \toprule
      \multirow{2}{*}{No.} & \multicolumn{2}{c}{ Resources} & \multicolumn{2}{c}{Constraints} & \multicolumn{3}{c}{Economic objectives } & \multicolumn{3}{c}{Environment objectives } & \multicolumn{2}{c}{Technical objectives } \\
                                   & $P_{d,i}$ (MW)  & $\varrho_i$ & $P_{i,\mathrm{min}}$ (MW)  & $P_{i,\mathrm{max}}$ (MW)  & $a_i^{\rm eco}$ & $b_i^{\rm eco}$ & $c_i^{\rm eco}$ & $a_i^{\rm fuel}$ & $b_i^{\rm fuel}$ & $c_i^{\rm fuel}$ & $a_i^{\rm tec}$ & $P_i^{\rm opt}$ \\
	  \toprule
      1  & 120 & 0.023 & 100 & 140 & 0.086 & 3.482 & 3.481 & 0.175 & 1.266 & 0.666 & 1.000 & 90 \\
      2  & 150 & 0.054 & 125 & 170 & 0.093 & 4.688 & 4.263 & 0.165 & 1.665 & 3.171 & 1.088 & 120 \\
      3  & 114 & 0.032 & 110 & 165 & 0.072 & 2.533 & 3.500 & 0.117 & 1.359 & 1.308 & 1.336 & 135 \\
      4  & 150 & 0.048 & 120 & 150 & 0.080 & 2.300 & 6.578 & 0.120 & 1.323 & 2.973 & 0.788 & 90 \\
      5  & 186 & 0.050 & 100 & 200 & 0.098 & 4.210 & 4.810 & 0.206 & 1.937 & 1.487 & 1.220 & 90 \\
      \bottomrule
    \end{tabular}
  \end{table*}   
Following \cite{Liu2021IEEE14, Xu2024Power, Liu2021LiNing,zhao2019distributed}, the proposed algorithm is validated on an IEEE-14 bus system with five generators using the $L_2$ preference index. The communication topology is represented by a circle, as shown in Fig. \ref{fig:microgrids}. The objectives for the $i$th generator are formulated as follows:
\begin{equation}
    \begin{aligned}
        \min &\ \left \{ f_i^{\rm eco}(P_i),f_i^{\rm env}(P_i),f_i^{\rm tec}(P_i)\right \}, \\
        {\rm s.t.}&\  \sum_{i=1}^{N}P_i=\sum_{i=1}^{N}(1+\varrho_i)P_{d,i},\\
        & \ P_{i,\mathrm{min}}\leq P_i \leq P_{i,\mathrm{max}},\\
    \end{aligned}
\end{equation}
where $P_i \in \mathbb{R}$ is the power output of the $i$th generator; $f_i^{\rm eco}(x_i) = a_i^{\rm eco}P_i^2 + b_i^{\rm eco}P_i + c_i^{\rm eco}$, $f_i^{\rm env}(x_i) =  r_t \left (a_i^{\rm fuel}P_i^2 + b_i^{\rm fuel}P_i + c_i^{\rm fuel} \right )$ and $f_i^{\rm tec}(x_i) =  a_i^{\rm tec} \left (P_i^2 - P_i^{\rm opt} \right )^2$ denote the economic, environmental, and technical objectives, respectively; $ P_{i,\mathrm{min}}$ and $P_{i,\mathrm{max}}$ are the lower and upper bounds only known to $i$th generator; $P_D=\sum_{i=1}^{N}P_{d,i}$ is the total demand ($P_{d,i} \in \mathbb{R}$ is the local resource data). The simulation parameters are displayed in TABLE \ref{table:parameters}, taken from the practical application in \cite{zhao2019distributed, li2021multiobjective} with modifications, and $r_t = 0.2$. The initial values of the power outputs are $\left [115; 150; 115; 145; 115\right ]$. The optimal power outputs of generators are $P_1^* = 137.830, P_2^* = 167.485, P_3^* = 165, P_4^* = 147.673, P_5^* =133.021$. This DCMRAP setup readily extends to more realistic scenarios (e.g., \cite{Liu2021LiNing, LI2022Expert, LIU2023Energy}). Without loss of generality, the following TBGs are used later: 
\begin{enumerate}
    \item TBG 1: $T(t,t_{\rm pre }) = \frac{\mathrm{d}}{\mathrm{d}t}\gamma(t)$ with
                \begin{equation}\label{TBG1}
                \gamma(t) =  \begin{cases} {\begin{array}{l} 
                \begin{aligned}
                    & 12t^2, & 0\leq t< t_{\rm pre}, \\ 
                    & t, & t \geq t_{\rm pre}.
                \end{aligned}
                \end{array}} \end{cases}
                \end{equation}
    \item TBG 2: $T(t,t_{\rm pre }) = \frac{\mathrm{d}}{\mathrm{d}t}\gamma(t)$ with
                \begin{equation}\label{TBG2}
                    \gamma(t) =  \begin{cases} {\begin{array}{l} 
                    \begin{aligned}
                        & 30t, & 0\leq t< t_{\rm pre}, \\ 
                        & t, & t \geq t_{\rm pre}.
                    \end{aligned}
                \end{array}} \end{cases}
                \end{equation}
    \item TBG 3 (\cite{NING2019Practicalfixed, GUO2022distributeddynamic, zhang_predefined-time_2023}): $T(t,t_{\rm pre }) =1+\frac{\dot{b}(t,t_{\rm pre})}{1-b(t,t_{\rm pre})+10^{-7}}$ with
                \begin{equation}\label{TBG3}
                \begin{aligned}
                    b(t,t_{\rm pre }) = \begin{cases} {\begin{array}{l} 
                        \begin{aligned}
                            & \frac{10}{t_{\rm pre }^6}t^6-\frac{24}{t_{\rm pre }^5}t^5+\frac{15}{t_{\rm pre }^4}t^4, &0\leq t< t_{\rm pre}, \\ 
                            & 1, &t \geq t_{\rm pre}.
                        \end{aligned}
                    \end{array}} \end{cases}
                \end{aligned}  
                \end{equation}
\end{enumerate}

\subsection{Basic Performance Test}\label{sectionA}
 To verify the efficacy of the proposed algorithm detailed in equations (\ref{solutionforsubproblem})-(\ref{dynamicETM2}), TBG 1 (\ref{TBG1}) is utilized. The parameters of the dynamic ETMs are chosen as $ \alpha_i^k = 10, \phi_i^k=0.1, \delta_i^k = 0.9,\beta_i^k = 0.1, \eta_i^k(0) = 500, \varsigma_i^{\rm eco} = 0.048, \varsigma_i^{\rm fuel} = 0.0551, \varsigma_i^{\rm tec} = 0.0551, \alpha_i = 10, \phi_i=0.05, \delta_i = 1, \beta_i = 0.1, \eta_i(0) = 800, \varsigma_i = 0.048$. The prescribed settling times are set as $t_{\rm pre 1} =t_{\rm pre 2}= 2$ and $t_{\rm pre 3} = 3$.
 
 

The simulation results are given in Fig. \ref{XEvent}(a), TABLE \ref{table:communications} and Case 1 in TABLE \ref{table:cases}. The power outputs, depicted in the upper section of Fig. \ref{XEvent}(a), demonstrate prescribed-time approximate convergence to the optimal solutions at $t_{\rm pre 3} = 3$ where the solid lines represent the actual outputs and the dashed lines represent the optimal solutions. From the output of the third generator, $P_3$, the differentiated projection operator is functional and effectively meets local constraints. Fig. \ref{W} shows that the weighting coefficients $\omega_i^k$ converge to optimal values within $t_{\rm pre 1}=2$ through an online learning process. This convergence effectively guides the preferences of agents while highlighting the decentralized nature of the system, eliminating the need for a central decision-maker. The blue curve in Fig. \ref{TBGE}(b) illustrates that power supply and demand are balanced, confirming that the global resource constraints are satisfied. The imbalance at $t_{\rm pre 3}$, quantified as $\sum_{i=1}^{N}(1+\varrho_i)P_{d,i} - P_i(t_{\rm pre 3})$, is a negligible $0.005$ MW, agaist the total demand of $751.008$ MW. According to (\ref{TBG1}), $T_3(t,t_{\rm pre 3}) = 1, \forall t >t_{\rm pre 3}$, so the convergence error diminishes as $t \rightarrow \infty$.

 Communication dynamics in the first five seconds, detailed in TABLE \ref{table:communications} and the lower part of Fig. \ref{XEvent}(a), reveal that inter-agent communication is discretized, with negligible exchanges after $t_{\rm pre 3}$.
 Fig. \ref{nueta} illustrates the evolution of the exchanged state $\bar{\nu}_i$ and the internal dynamic variables $\eta_i$, further validating the effectiveness of the proposed dynamic ETMs.

\begin{figure*}[ht!]
  \centering
  \subfloat[]{\includegraphics[width=0.245\textwidth]{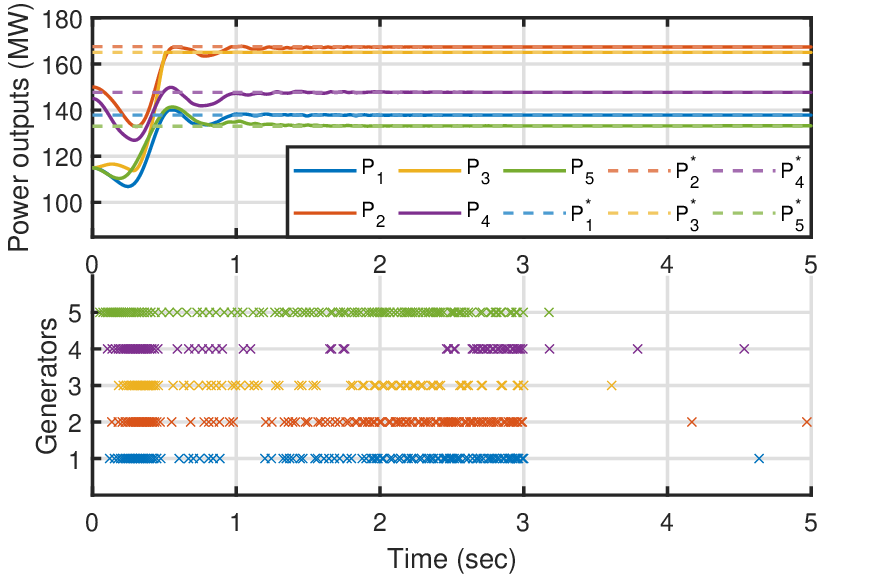}}
  \subfloat[]{\includegraphics[width=0.245\textwidth]{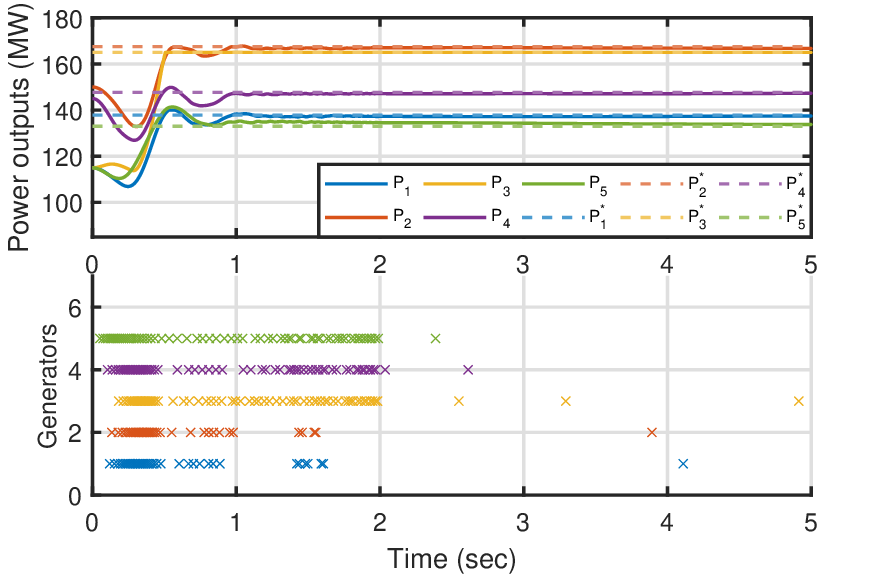}}
  \subfloat[]{\includegraphics[width=0.245\textwidth]{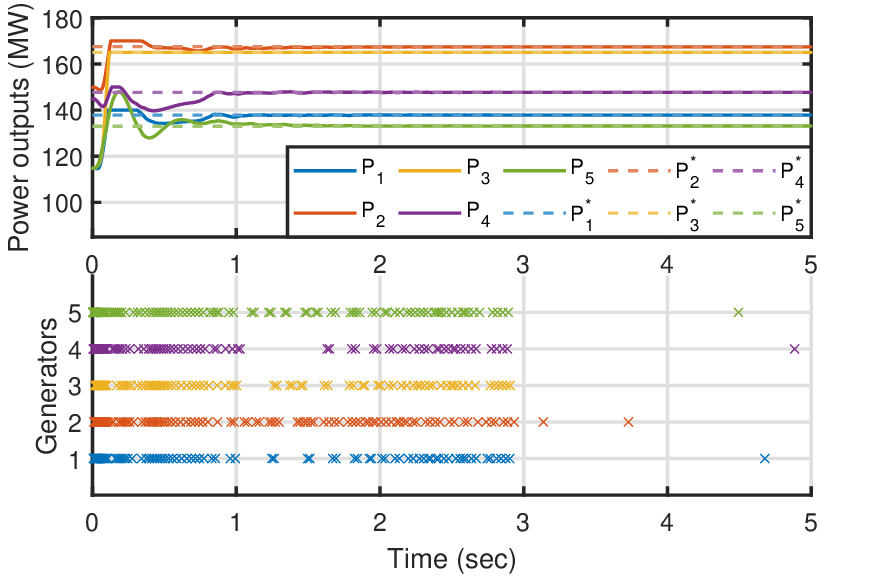}}
  \subfloat[]{\includegraphics[width=0.245\textwidth]{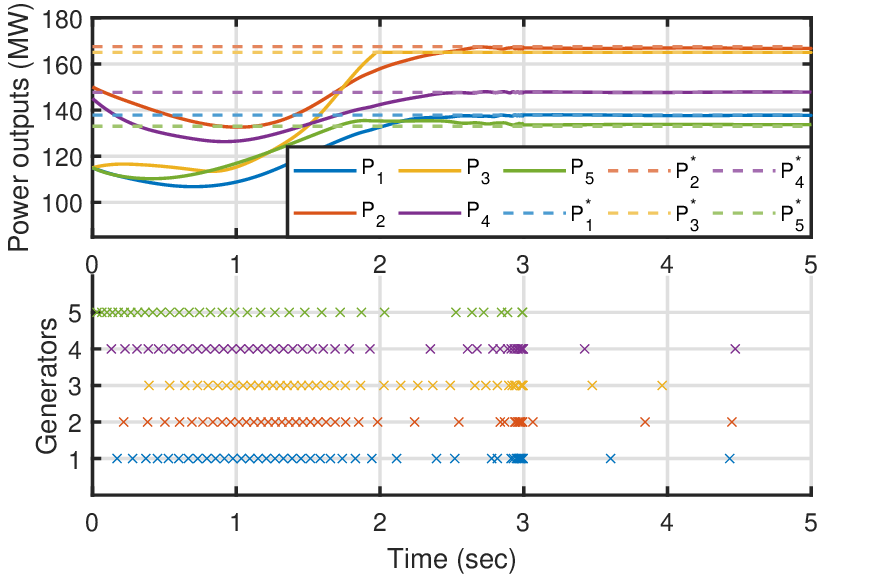}}
  \caption{Power outputs and communications of $u_i$ using the proposed algorithm (\ref{solutionforsubproblem})-(\ref{dynamicETM2}) with different TBGs. (a) TBG 1 (\ref{TBG1}) with $t_{\rm pre 1} =t_{\rm pre 2}= 2$, $t_{\rm pre 3} = 3$; (b) TBG 1 (\ref{TBG1}) with $t_{\rm pre 1} =t_{\rm pre 2}= 1$, $t_{\rm pre 3} = 2$; (c) TBG 2 (\ref{TBG2}) with $t_{\rm pre 1} =t_{\rm pre 2}= 2$, $t_{\rm pre 3} = 3$; (d) TBG 3 (\ref{TBG3}) with $t_{\rm pre 1} =t_{\rm pre 2}= 2$, $t_{\rm pre 3} = 3$.}
  \label{XEvent}
\end{figure*}

\begin{figure}
    \centering
    \includegraphics[width=0.7\linewidth]{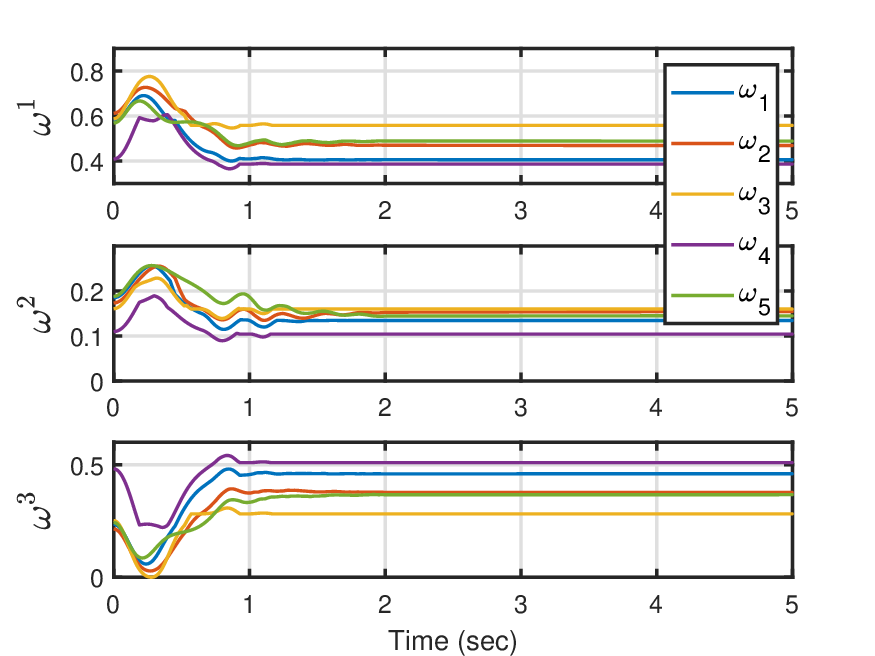}
    \caption{Evolution of weighting coefficients in Case 1.}
    \label{W}
\end{figure}

\begin{figure}[ht!]
  \centering
  \subfloat[]{\includegraphics[width=0.48\linewidth]{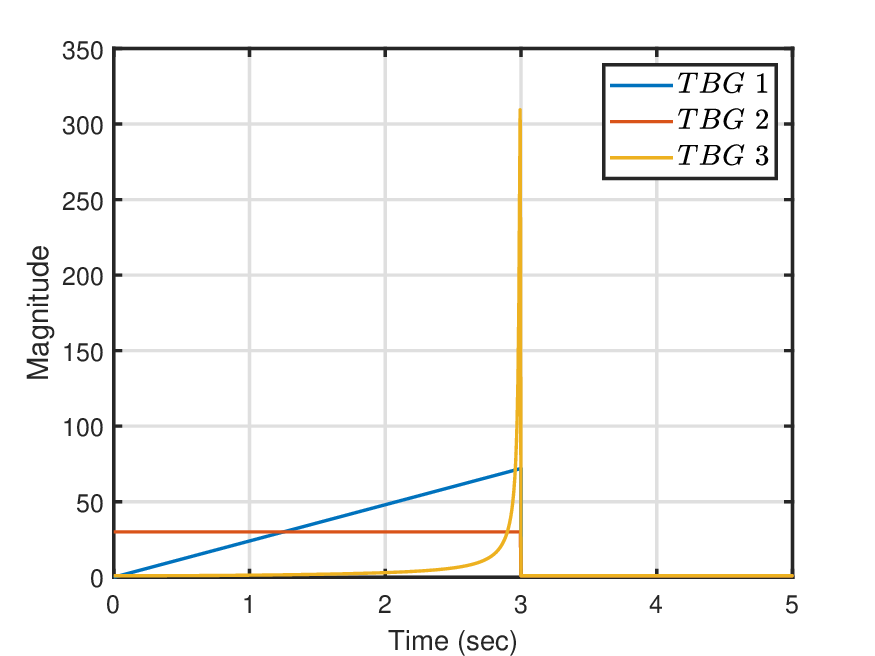}}
  \subfloat[]{\includegraphics[width=0.48\linewidth]{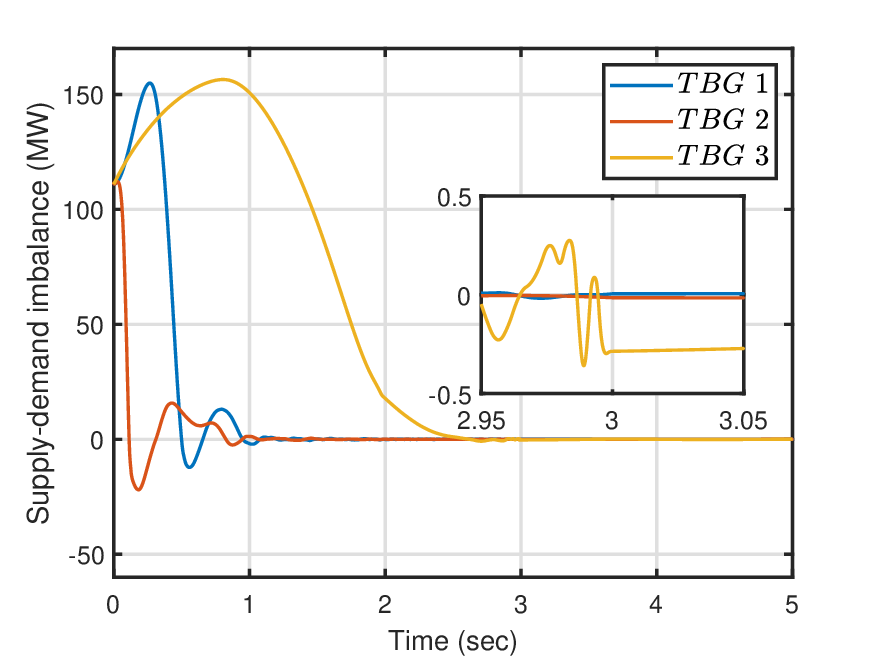}}
  \caption{Simulation results using different TBGs (\ref{TBG1})-(\ref{TBG3}). (a) Outputs of TBGs; (b) Imbalance between demand and supply $\sum_{i=1}^{N}(1+\varrho_i)P_{d,i} - P_i$. }
    \label{TBGE}
\end{figure}

\begin{figure}[ht!]
  \centering
  \subfloat[]{\includegraphics[width=0.48\linewidth]{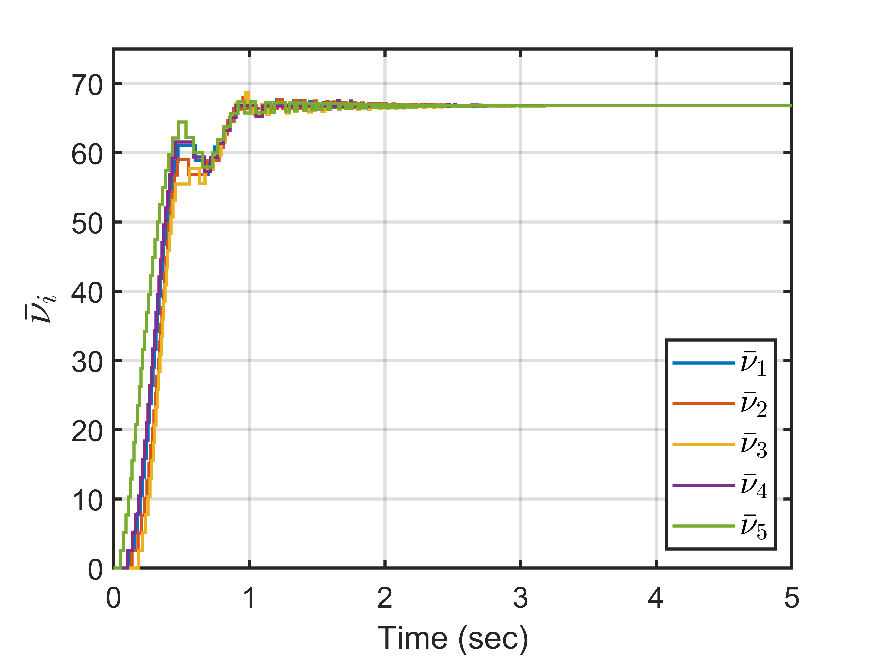}}
  \subfloat[]{\includegraphics[width=0.48\linewidth]{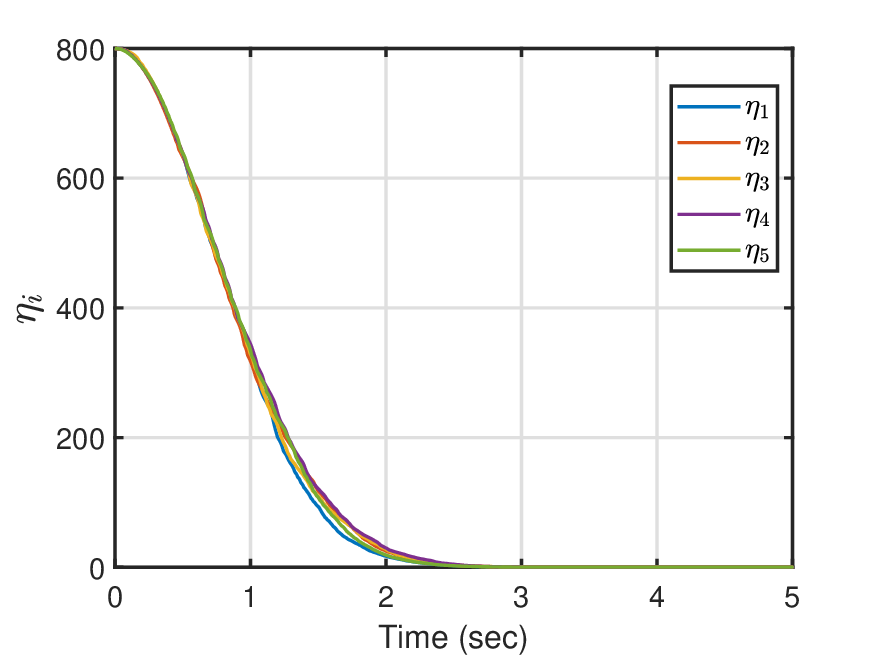}}
  \caption{Dynamic evolution of (a) Exchanged information $\bar{\nu}_i$; (b) Internal dynamic variables $\eta_i$.} 
    \label{nueta}
\end{figure}

 \begin{table}[ht!]
    \centering
    \caption{Communication Results for Case 1.}
    \label{table:communications}
    \begin{tabularx}{\linewidth}{>{\centering\arraybackslash}X 
                                  >{\centering\arraybackslash}X 
                                  >{\centering\arraybackslash}X 
                                  >{\centering\arraybackslash}X 
                                  >{\centering\arraybackslash}X 
                                  >{\centering\arraybackslash}X}
      \toprule
      \multirow{2}{*}{Objective} & \multicolumn{5}{c}{ Communications$^{\mathrm{a}}$}  \\
                                & G$1$  & G$2$  & G$3$  & G$4$ & G$5$  \\ 
	  \toprule
      $f_i^{\rm eco}$  & 92 & 93 & 97 & 95 & 82  \\
      $f_i^{\rm env}$  & 117 & 114 & 131 & 139 & 119   \\
      $f_i^{\rm tec}$  & 63 & 61 & 87 & 83 & 58   \\
      $u_i$            & 107 & 113 & 76 & 67 & 115    \\
      \hline
      Total            & 379 & 381 & 391 & 384 & 374 \\
      \bottomrule
      \multicolumn{6}{p{210pt}}{$^{\mathrm{a}}$Communications: The number of triggered events in 5 seconds. }\\
    \end{tabularx}
  \end{table}

\begin{table*}[ht!]
    \centering
    \caption{Simulation Results in Different Cases.}
    \label{table:cases}
    \begin{tabular}{ccccccccccccc}
      \toprule
      \multirow{2}{*}{Cases}  &\multirow{2}{*}{TBGs} & \multicolumn{3}{c}{ Settling times (s)} & \multirow{2}{*}{ETMs}  & \multicolumn{6}{c}{ Communications$^{\mathrm{a}}$}  &\multirow{2}{*}{CE$^{\mathrm{b}}$ (MW)} \\
                             &  &$t_{\rm pre 1}$ &$t_{\rm pre 2}$ &$t_{\rm pre 3}$ & & G$1$  & G$2$  & G$3$  & G$4$ & G$5$ & Total &\\ 
	  \toprule
      Case 1  & TBG 1 & 2 & 2 & 3 & (\ref{thresholdETM1})-(\ref{dynamicETM}) and (\ref{thresholdETM2})-(\ref{dynamicETM2})  & 379 & 381 & 391 & 384 & 374 & 1909 & 0.005 \\
      Case 2  & TBG 1 & 1 & 1 & 2 & (\ref{thresholdETM1})-(\ref{dynamicETM}) and (\ref{thresholdETM2})-(\ref{dynamicETM2})  & 230 & 237 & 266 & 266 & 261 & 1260 & 0.394 \\
      Case 3  & TBG 2 & 2 & 2 & 3 & (\ref{thresholdETM1})-(\ref{dynamicETM}) and (\ref{thresholdETM2})-(\ref{dynamicETM2})  & 387 & 406 & 440 & 444 & 403 & 2080 & -0.015 \\
      Case 4  & TBG 3 & 2 & 2 & 3 & (\ref{thresholdETM1})-(\ref{dynamicETM}) and (\ref{thresholdETM2})-(\ref{dynamicETM2})  & 248 & 253 & 251 & 269 & 236 & 1257 & -0.286 \\
      Case 5  & TBG 3 & 2 & 2 & 3 & (\ref{thresholdETM1})-(\ref{dynamicETM}) and (\ref{thresholdETM2})-(\ref{dynamicETM2})  & 268 & 271 & 269 & 287 & 263 & 1358 & 0.149 \\
      Case 6  & TBG 1 & 2 & 2 & 3 & (\ref{thresholdETMstatic}) \cite{gao2022ETMGT, zhao2019distributed, wu2020event, chai2024event} & 8850 & 9742 & 9147 & 8598 & 9751 & 46088 & -0.000256 \\
      Case 7  & TBG 1 & 2 & 2 & 3 & (\ref{thresholdETMdynamicguo})-(\ref{dynamicETMguo}) \cite{GUO2022distributeddynamic} & 378 & 412 & 431 & 427 & 397 & 2045 & 0.00904 \\
      \bottomrule
      \multicolumn{9}{p{250pt}}{$^{\mathrm{a}}$Communications: Total number of events triggered in the first 5 seconds. }\\
      \multicolumn{9}{p{300pt}}{$^{\mathrm{b}}$CE: Convergence error at the prescribed settling time: $\sum_{i=1}^{N}(1+\varrho_i)P_{d,i} - P_i(t_{\rm pre 3})$. }\\
    \end{tabular}
  \end{table*}

\subsection{Prescribed-Time Control Validation}
To further validate the prescribed-time approximate convergence achieved by the proposed algorithm, the settling times are set as $t_{\rm pre 1} = t_{\rm pre 2} = 1$ and $t_{\rm pre 3} = 2$, with all other conditions and parameters identical to Section \ref{sectionA}. The results are presented in Fig. \ref{XEvent}(b) and Case 2 in Table \ref{table:cases}, which clearly illustrate that the power outputs converge to the optimal solutions by $t_{\rm pre 3} = 2$ under event-triggered communication. Combined with the results from Case 1, this demonstrates that the proposed algorithm ensures prescribed-time approximate convergence with arbitrarily defined settling times, independent of initial conditions or other parameters. This highlights the robustness and applicability of the algorithm.

\subsection{Different TBGs Comparison}
To further evaluate the flexibility of our proposed algorithm, TBG 2 \cite{liu2023multiobjective} and TBG 3 \cite{NING2019Practicalfixed, GUO2022distributeddynamic, zhang_predefined-time_2023} are implemented for comparison studies. All other settings are are identical to Case~1. The experimental scenarios are detailed in Cases 3 and 4 of TABLE \ref{table:cases}, with corresponding curves depicted in Fig. \ref{XEvent}(c), (d), and Fig. \ref{TBGE}.

The power output data in Fig. \ref{XEvent} shows that, despite variations in output curves and convergence errors, all tested configurations achieve prescribed-time approximate convergence at $t_{\rm pre 3}$. Notably, the dynamic responses vary with the different TBGs: TBG 2 (Case 3) exhibits the highest initial rise and overshoot, TBG 1 (Case 1) is moderate, and TBG 3 (Case 4) shows the mildest initial response. The output profiles of TBGs corroborate this in Fig. \ref{TBGE}(a), where TBG 2 shows the highest initial output effort, whereas the output of TBG 3 increases most gradually at the beginning. 

As shown in the lower sections of Fig. \ref{XEvent} and TABLE \ref{table:cases}, Case 3 has the highest communication frequency, while Case 4 has the lowest. From Fig. \ref{TBGE}(b) and TABLE \ref{table:cases}, TBG 1 achieves the smallest convergence error, while TBG 3 has the largest error, despite its advantages in controller output smoothness and reduced communication demands. Even when the tolerance parameter is tightened from $10^{-7}$ to $10^{-9}$ (Case 5), the error remains larger than the other two TBGs, with increased communication frequency. This indicates that reducing the error tolerance improves convergence accuracy but compromises communication efficiency and requires higher system precision.

In conclusion, the simulation results validate the practicality and flexibility of our algorithm, highlighting its ability to tailor TBG configurations to meet specific performance criteria, as supported by Remarks \ref{remark4} and \ref{remark5}.

\subsection{ETMs Comparison Study}
To further evaluate the communication efficiency of the proposed algorithm, a comparison study is conducted, documented as Cases 6 and 7. In Case 6, the static ETM (\ref{thresholdETMstatic}) from \cite{gao2022ETMGT, zhao2019distributed, wu2020event, chai2024event} replaces (\ref{thresholdETM1})–(\ref{dynamicETM}) and (\ref{thresholdETM2})–(\ref{dynamicETM2}), while in Case 7, the dynamic ETM (\ref{thresholdETMdynamicguo})–(\ref{dynamicETMguo}) from \cite{GUO2022distributeddynamic} is used. For fairness, all parameters are identical to Case~1, except for the ETM-specific choices stated in Remark~\ref{remark7} (i.e., $\alpha_i^k, \alpha_i \to \infty$ in Case~6 and $\beta_i^k, \beta_i=0$ in Case~7). Results are presented in TABLE \ref{table:cases} and Fig. \ref{ETMs}, with a bar chart in Fig. \ref{ETMs}(b) providing a statistical comparison of communication events across different ETMs, showing that Case 1 has less communication than other two Cases. Specifically, from \ref{ETMs}(b), the proposed dynamic ETM reduces communication events by $96.18\%$ compared to the static ETM (\ref{thresholdETMstatic}) and by $12.97\%$ compared to the dynamic ETM in \cite{GUO2022distributeddynamic}. These results support Remark \ref{remark7} and validate the superior efficiency of the proposed ETMs in minimizing communication overhead while maintaining robust system performance.

\begin{figure}[ht!]
  \centering
  \subfloat[]{\includegraphics[width=0.48\linewidth]{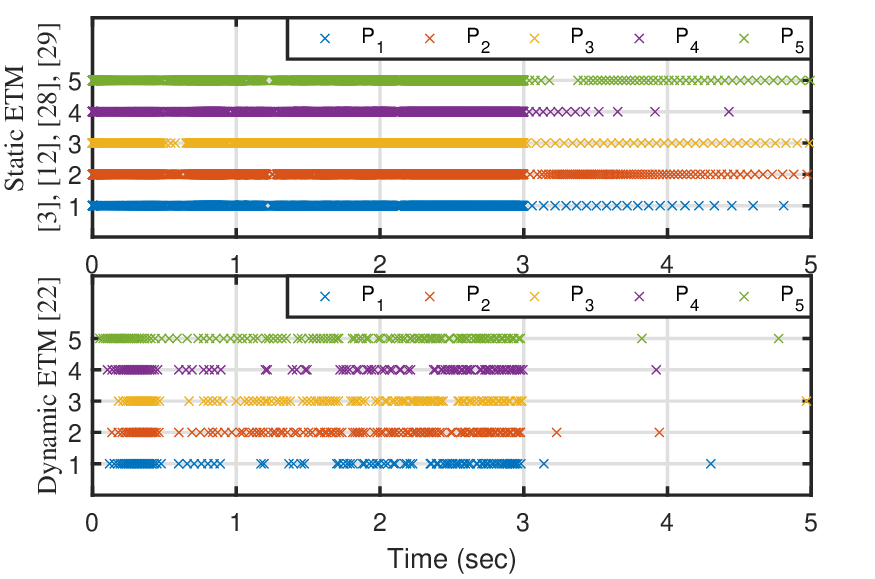}}
  \subfloat[]{\includegraphics[width=0.48\linewidth]{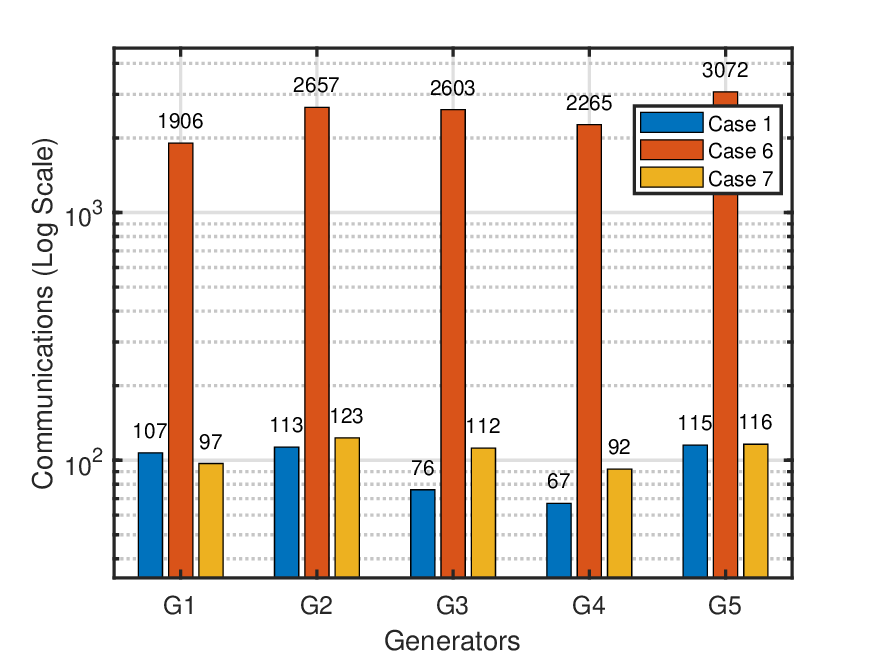}}
  \caption{Simulation results: (a) Event times of $ u_i $ using the static ETM (\ref{thresholdETMstatic}) and the dynamic ETM (\ref{thresholdETMdynamicguo})-(\ref{dynamicETMguo}); (b) Number of communications for different ETMs in the first $5$ seconds.}
    \label{ETMs}
\end{figure}

\section{Conclusion}\label{section5}
This paper addresses CMRAPs across multi-agent networks under both local and global constraints. By reformulating CMRAPs into a weighted $L_p$-based framework, our study finds the optimal compromised solution and avoids the reliance on predetermined weighting coefficients or centralized decision-making authority. Central to our approach is the novel distributed algorithm characterized by prescribed-time control and dynamic ETMs. The generalized TBGs provide the flexibility to tailor the control performance to the demands of various operational scenarios. Furthermore, our dynamic ETMs significantly reduce communication overhead by incorporating a network-based error term and interacting with TBGs. The feasibility and efficiency of the algorithm are validated under mild assumptions through rigorous Lyapunov stability analysis and detailed simulations on a microgrid system. Comparative studies highlight its superior performance over existing methods. The algorithm effectively addresses complex multi-objective optimization tasks on low-cost, resource-constrained agents with minimal communication and computational resources, advancing distributed optimization for more practical and complex real-world applications.

\bibliographystyle{IEEEtran}
\bibliography{ref}

\begin{IEEEbiography}[{\includegraphics[width=1in,height=1.25in,clip,keepaspectratio]{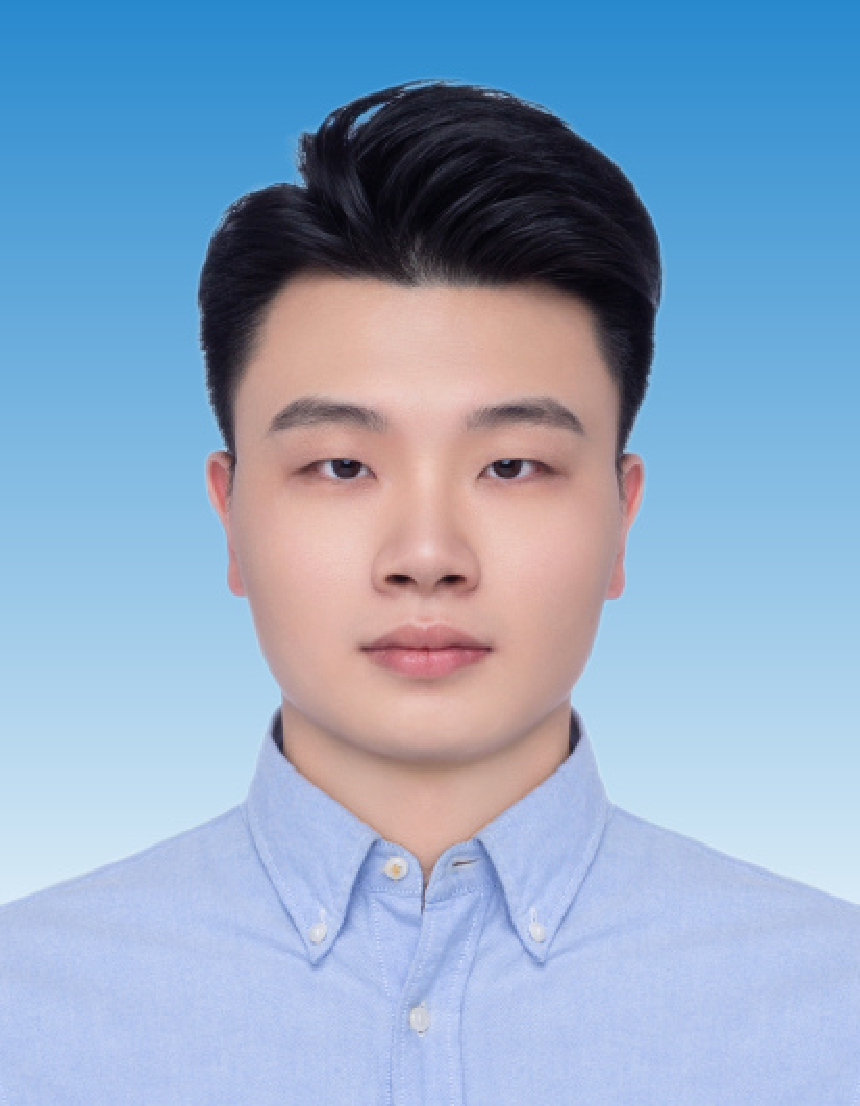}}]%
{Tengyang Gong} received the M.Sc. degree in advanced control and systems engineering from the Department of Electrical and Electronic Engineering, The University of Manchester, U.K., in 2022, where he is currently pursuing the Ph.D. degree in electrical and electronic engineering.

His research interests include the distributed optimization, control, and game theory of multiagent systems.
\end{IEEEbiography}
\begin{IEEEbiography}[{\includegraphics[width=1in,height=1.25in,clip,keepaspectratio]{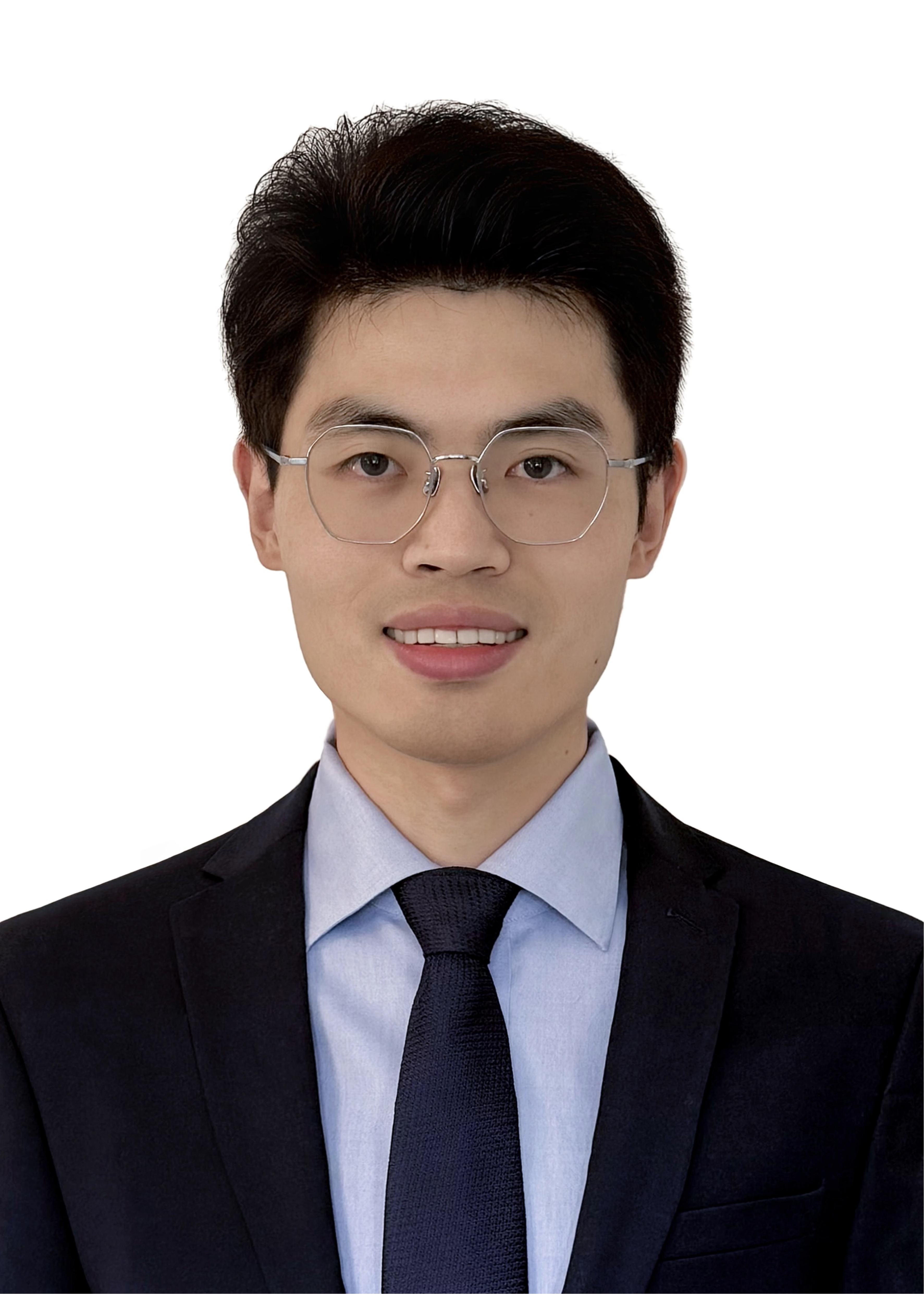}}]%
{Zhongguo Li}  received the B.Eng. and Ph.D. degrees in electrical and electronic engineering from The University of Manchester, Manchester, U.K., in 2017 and 2021, respectively.

He was a Lecturer with University College London, London, U.K., and a Research Associate at Loughborough University, Loughborough, U.K. He is currently a Lecturer (an Assistant Professor) in robotics, control, communication, and AI with The University of Manchester. His research interests include optimization and decision-making for advanced control, distributed algorithm development for game and learning in network-connected multiagent systems, and their applications in robotics and autonomous vehicles.
\end{IEEEbiography}
\begin{IEEEbiography}[{\includegraphics[width=1in,height=1.25in,clip,keepaspectratio]{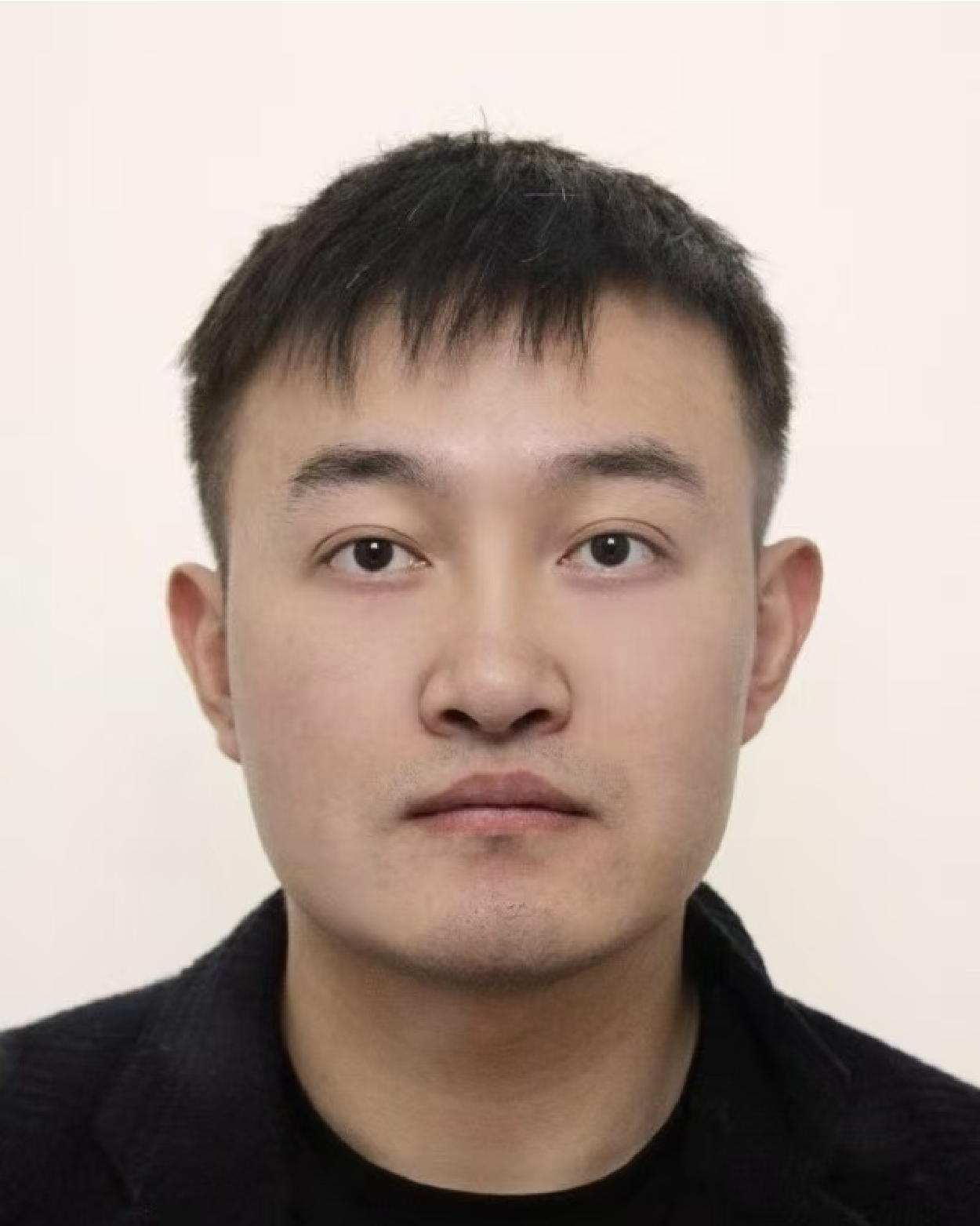}}]%
{Yiqiao Xu} received the M.Sc. degree in advanced control and systems engineering and the Ph.D. degree in electrical and electronic engineering from The University of Manchester, U.K., in 2018 and 2023, respectively.

Since 2023, he has been a Post-Doctoral Research Associate with the Power and Energy Division, The University of Manchester. His research interests include distributed optimization and learning-based control of power networks and multienergy systems.
\end{IEEEbiography}
\begin{IEEEbiography}[{\includegraphics[width=1in,height=1.25in,clip,keepaspectratio]{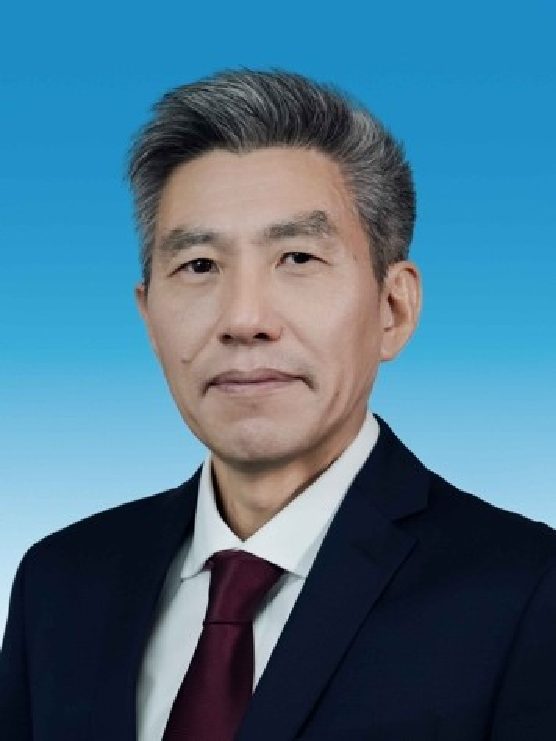}}]%
{Zhengtao Ding}  received the B.Eng. degree in thermal energy from Tsinghua University, Beijing, China, in 1984, and the M.Sc. degree in systems and control and the Ph.D. degree in control systems from the University of Manchester Institute of Science and Technology, Manchester, U.K., in 1986 and 1989, respectively. After working in Singapore for 10 years, he joined The University of Manchester, Manchester, in 2003, where he is currently a Professor of Control Systems. He has authored/co-authored five books, including the book Nonlinear and Adaptive Control Systems (IET, 2013), and has published over 400 research articles. His research interests include nonlinear and adaptive control theory and their applications, more recently, network-based control, distributed optimization, and distributed learning, with applications to power systems and robotics.

Prof. Ding is a member of the IEEE Technical Committee on Nonlinear Systems and Control, the IEEE Technical Committee on Intelligent Control, and the IFAC Technical Committee on Adaptive and Learning Systems. He was elected as a fellow of The Alan Turing Institute in 2021, the U.K.’s national institute for data science and artificial intelligence. He serves/has served as the Editor-in-Chief for Drones and Autonomous Vehicles, the Subject Chief Editor for Nonlinear Control for Frontiers, and an Associate Editor for Scientific Reports, IEEE TRANSACTIONS ON AUTOMATIC CONTROL, IEEE TRANSACTIONS ON CIRCUITS AND SYSTEMS—II: EXPRESS BRIEFS, IEEE CONTROL SYSTEMS LETTERS, Transactions of the Institute of Measurement and Control, Control Theory and Technology, Unmanned Systems, and several other journals.
\end{IEEEbiography}

\end{document}